\documentclass[a4paper,UKenglish,cleveref, autoref, thm-restate]{lipics-v2021}
\usepackage[colorinlistoftodos]{todonotes}
\nolinenumbers
\usepackage{algorithmic}
\usepackage[ruled,linesnumbered]{algorithm2e}
\usepackage{amsmath}

\usepackage{xspace}
\usepackage{tcolorbox}

\newcommand{\icf}{{independence covering family}}

\DeclareMathAlphabet{\w}{OT1}{pzc}{m}{it}

\newcommand{\C}[1]{\ensuremath{\operatorname{\mathcal {#1} }}}

\newcommand{\F}{{\mathcal F}}
\newcommand{\OO}{{\mathcal O}}

\newcommand{\FPT} {{\sf FPT}\xspace}

\crefname{claim}{claim}{claims}
\crefname{fact}{fact}{facts}

\newcommand{\h}[1]{\end{document}}

\usepackage{boxedminipage}
 
\usepackage{bold-extra}
\usepackage{lmodern}
\usepackage[T1]{fontenc}
\usepackage{textcomp}

\newcommand{\ifam}[2]{\ensuremath{{\mathscr F}(#1,#2)}}

 \newcommand{\defparproblem}[4]{
\noindent\fbox{
  \begin{minipage}{0.96\textwidth}
   \begin{tabular*}{\textwidth}{@{\extracolsep{\fill}}lr} #1  & {\bf{Parameter:}} #3
 \\ \end{tabular*}
  {\bf{Input:}} #2  \\
  {\bf{Question:}} #4
  \end{minipage}
  }
 \vspace{1mm}
}

 \newcommand{\defproblem}[3]{
\noindent\fbox{
  \begin{minipage}{0.96\textwidth}
   \begin{tabular*}{\textwidth}{@{\extracolsep{\fill}}lr} #1  
 \\ \end{tabular*}
  {\bf{Input:}} #2  \\
  {\bf{Question:}} #3
  \end{minipage}
  }
 \vspace{1mm}
}

\newcommand{\what}{\widehat}
\newcommand{\wtilde}{\widetilde}
\newcommand{\pgc}{{\textsc{Partial Grundy Coloring}}\xspace}
\newcommand{\KijGrundy}{\textsc{$K_{i,j}$-Free Grundy Coloring}\xspace}
\newcommand{\Grundycol}{\textsc{Grundy Coloring}\xspace}
\newcommand{\qrep}[1]{\subseteq^{{#1}}_{grep}}
\newcommand{\sm}[1]{{\sf sum}(#1)}
\newcommand{\wvec}{\overrightarrow}

\newcommand{\degval}{i}
\newcommand{\etaval}[1]{\degval \cdot \sm{#1} \cdot k}
\newcommand{\szsize}{(p \cdot \eta)^{i+1}} 
\newcommand{\degr}{d}
\newcommand{\alphaval}{3 \cdot k \cdot (p \eta)^{i+1}}

\newcommand{\longversion}[1]{#1}
\newcommand{\shortversion}[1]{}

\newcommand{\pgcstruct}{{\sc Structural Partial Grundy Coloring~}}
\newcommand{\pgcstructshort}{{\sc SPGC~}}

\title{Parameterized Saga of First-Fit and Last-Fit Coloring}

\author{Akanksha Agrawal}{Department of Computer Science and Engineering, Indian Institute of Technology Madras, India}{akanksha@cse.iitm.ac.in}{}{}
\author{Daniel Lokshtanov}{Department of Computer Science, University of California Santa Barbara, USA}{daniello@ucsb.edu}{}{}
\author{Fahad Panolan}{School of Computing, University of Leeds, UK}{F.Panolan@leeds.ac.uk}{}{}
\author{Saket Saurabh}{Theoretical Computer Science, The Institute of Mathematical Sciences, HBNI,  India\\ Department of Informatics, University of Bergen, Norway}{saket@imsc.res.in}{}{}
\author{Shaily Verma}{Algorithm Engineering Group, Hasso Plattner Institute, Germany }{Shaily.Verma@hpi.de}{}{}

\authorrunning{Agrawal et al.} 

\Copyright{Anonymous} 



\category{} 

\begin{document}

\maketitle

\begin{abstract}
The classic greedy coloring (first-fit) algorithm considers the vertices of an input graph $G$ in a given order and assigns the first available color to each vertex $v$ in $G$. In the {\sc Grundy Coloring} problem, the task is to find an ordering of the vertices that will force the greedy algorithm to use as many colors as possible. 
In the {\sc Partial Grundy Coloring}, the task is also to color the graph using as many colors as possible. This time, however, we may select both the ordering in which the vertices are considered and which color to assign the vertex. 
The only constraint is that the color assigned to a vertex $v$ is a color previously used for another vertex if such a color is available. 
{Partial Grundy coloring of a graph corresponds to the vertex coloring  produced by the greedy coloring heuristics called \emph{last-fit} coloring, which assigns the last available color to each vertex in the given order.}

Whether {\sc Grundy Coloring} and {\sc Partial Grundy Coloring} admit fixed-parameter tractable (FPT) algorithms, algorithms with running time $f(k)n^{\OO(1)}$, where $k$ is the number of colors, was posed as an open problem by Zaker and by Effantin et al., respectively.

Recently, Aboulker et al. (STACS 2020 and Algorithmica 2022) resolved the question for \Grundycol\ in the negative by showing that the problem is W[1]-hard. For {\sc Partial Grundy Coloring}, they obtain an FPT algorithm on graphs that do not contain $K_{i,j}$ as a subgraph (a.k.a. $K_{i,j}$-free graphs). Aboulker et al.~re-iterate the question of whether there exists an FPT algorithm for {\sc Partial Grundy Coloring} on general graphs and also asks whether {\sc Grundy Coloring} admits an FPT algorithm on $K_{i,j}$-free graphs. We give FPT algorithms for {\sc Partial Grundy Coloring} on general graphs and for {\sc Grundy Coloring} on $K_{i,j}$-free graphs, resolving both the questions in the affirmative. We believe that our new structural theorems for partial Grundy coloring and ``representative-family'' like sets for $K_{i,j}$-free graphs that we use in obtaining our results may have wider algorithmic applications.

\end{abstract}

\newpage
\setcounter{page}{1}

\section{Introduction}
%
%
%
A {\em proper coloring} of a graph $G$ is an assignment of colors to its vertices such that none of the edges have endpoints of the same color. In {\sc $k$-Coloring}, we are given a graph $G$, and the objective is to test whether $G$ admits a proper coloring using at most $k$ colors.
The {\sc $k$-Coloring} problem is one of the classical NP-hard problems, and it is NP-complete for every fixed $k \geq 3$. The problem is notoriously hard even to approximate. Indeed, approximating {\sc $k$-Coloring} within $\C{O}(n^{1-\epsilon})$, for any $\epsilon >0$, is hard~\cite{DBLP:journals/jcss/FeigeK98}. 
Also, under a variant of Unique Games Conjecture, there is no constant factor approximation for {\sc $3$-Coloring}~\cite{DBLP:journals/siamcomp/DinurMR09}.

The {\sc $k$-Coloring} problem has varied applications ranging from scheduling, register allocations, pattern matching, and beyond~\cite{DBLP:journals/cl/ChaitinACCHM81,chen1998unifying,marx2004graph}.
Owing to this,
several heuristics-based algorithms have been developed for the problem.
One of the most natural greedy strategies considers the vertices of an input graph $G$ in an arbitrary order and assigns to each vertex the first available
color in the palette (the color palette for us is $\mathbb{N}$). In literature, this is called the {\em first-fit} rule. Notice that there is nothing special about using the ``first'' available color; one may instead opt for any of the previously used colors, if available, before using a new color; let us call this greedy rule the {\em any-available} rule. It leads to yet another greedy strategy to properly color a graph, and one can easily prove that this greedy strategy is equivalent to the ``last-(available) fit'' rule.

For any greedy strategy, one may wonder: {\em How bad can the strategy perform for the given instance?} The above leads us to the well-studied fundamental combinatorial problems, \Grundycol\ and \pgc, that arise from the aforementioned greedy strategies for proper coloring. In the \Grundycol\ problem, we are given a graph $G$ on $n$ vertices and an integer $k$, and the goal is to check if there is an ordering of the vertices on which the first-fit greedy algorithm for proper coloring uses at least $k$ colors?
Similarly, we can define the \pgc\ problem, where the objective is to check if, for the given graph $G$ on $n$ vertices and integer $k$, there is an ordering of the vertices on which the any-available greedy algorithm uses at least $k$ colors. In this paper, we consider these two problems in the realm of parameterized complexity.



The \Grundycol\ problem has a rich history dating back to 1939~\cite{grundy1939mathematics}. Goyal and Vishwanathan~\cite{goyal1997np} proved that \Grundycol\ is NP-hard. Since then, there has been a flurry of results on the computational and combinatorial aspects of the problem both on general graphs and on restricted graph classes, see, for instance~\cite{havet2013grundy,zaker2005grundy,verma2020grundy,sampaio2012algorithmic,DBLP:journals/dam/BonnetFKS18,erdos2003equality,DBLP:journals/dam/KiersteadS11,christen1979some,DBLP:journals/dam/Zaker07a,DBLP:journals/jco/TangWHZ17,DBLP:journals/endm/AraujoS09,kortsarz2007lower,hedetniemi1982linear,zaker2006results,belmonte2020grundy} (this list is only illustrative, not comprehensive). The problem \pgc\ was introduced by Erd{\"o}s et al.~\cite{erdos2003equality} and was first shown to be NP-hard by Shi et al.~\cite{shi2005algorithm}. The problem has gained quite some attention thereafter; see, for instance~\cite{aboulker2022grundy,effantin2016characterization,panda2019partial,IBIAPINA2020111920,verma2020grundy,DBLP:conf/ispa/DekarEK07,effantin2016characterization,DBLP:journals/gc/BalakrishnanK11,balakrishnan2013interpolation,verma2020grundy}.



These problems have also been extensively studied from the parameterized complexity perspective. Unlike {\sc $k$-Coloring}, both these problems admit XP algorithms~\cite{zaker2006results,effantin2016characterization}, i.e., an algorithm running in time bounded by $|V(G)|^{f(k)}$. The above naturally raises the question of whether they are fixed-parameter tractable (FPT), i.e., admit an algorithm running in time $f(k)\cdot |V(G)|^{\C{O}(1)}$. In fact, these problems have also been explicitly stated as open problems~\cite{aboulker2022grundy,DBLP:journals/dam/BonnetFKS18,havet2013grundy,sampaio2012algorithmic}. 

Havet and Sampaio~\cite{havet2013grundy} studied \Grundycol with an alternate parameter and showed that the problem of testing whether there is a Grundy coloring with at least $|V(G)|-q$ colors is FPT parameterized by $q$. Bonnet et al.~\cite{DBLP:journals/dam/BonnetFKS18} initiated a systematic study of designing parameterized and exact exponential time algorithms for \Grundycol\ and obtained FPT algorithms for the problem for several structured graph classes. They gave an exact algorithm for \Grundycol running in time $2.443^n\cdot n^{\C{O}(1)}$ and also showed that the problem is \FPT on chordal graphs, claw-free graphs and graphs excluding a fixed minor. In the same paper, they stated the tractability status of \Grundycol on general graphs parameterized by the treewidth or the number of colors as central open questions. Belmonte et al.~\cite{belmonte2020grundy} resolved the first question by proving that \Grundycol is W[1]-hard parameterized by treewidth, but surprisingly, it becomes FPT parameterized by pathwidth. Later, Aboulker et al.~\cite{aboulker2022grundy} proved that \Grundycol\ does not admit an FPT algorithm (parameterized by the number of colors) and obtained an FPT algorithm for \pgc~on $K_{t,t}$-free graphs (which includes graphs of bounded degeneracy, graphs excluding some fixed graph as minor/topological minors, graphs of bounded expansion and no-where dense graphs).
A graph is $K_{i,j}$-free if it does not have a {\em subgraph} isomorphic to the complete bipartite graph with $i$ and $j$ vertices, respectively, on the two sides. 
They concluded their work with the following natural open questions:
\vspace{-1mm}
\begin{tcolorbox}[colback=blue!5!white,colframe=white!100!black]
\begin{description}
\item[Question 1:] Does \pgc\ admit an FPT algorithm?
\label{abc}
\item[Question 2:] Does \Grundycol\ admit an FPT algorithm on $K_{i,j}$-free graphs?
\end{description}
\end{tcolorbox}

In this paper, we resolve the questions $1$ and $2$ in the affirmative by a new structural result and a new notion of representative families for $K_{i,j}$-free graphs, respectively. In the next section, we give an intuitive overview of both results, highlighting our difficulties and our approaches to overcome them. 

\subsection{Our Results, Methods and Overview}
Our first result is the following.
\begin{theorem}
\label{thm:pgc}
\pgc\ is solvable in time $2^{\OO(k^5)} \cdot n^{\OO(1)}$.
\end{theorem}

Our algorithm starts with the known ``witness reformulation'' of \pgc. It is known that $(G,k)$ is a yes-instance of \pgc if and only if there is a vertex subset $W$ of size at most $k^2$ such that, $(G[W],k)$ is a yes-instance of the problem. 
In the above, the set $W$ is known as a {\em small witness set}. Our algorithm is about finding such a set $W$ of size at most $k^2$. Observe that this witness reformulation immediately implies that \pgc admits an algorithm with running time $n^{\OO(k^2)}$ time. To build our intuition, we first give a simple algorithm for the problem on graphs of bounded degeneracy (or even, no-where dense graphs). This algorithm has two main steps: (a) classical color-coding of Alon-Yuster-Zwick~\cite{DBLP:journals/jacm/AlonYZ95}, and (b) independence covering lemma of Lokshtanov et al.~\cite{lokshtanov2020covering}. 

Let $(G,k)$ be a yes instance of \pgc, where $G$ is a $d$-degenerate graph on $n$ vertices, and $W$ be a small witness set of size at most $k^2$. As $(G[W],k)$ must be a yes-instance of the problem, there exists an ordering of the vertices such that when we apply {\em any-available} greedy rule, it uses at least $k$ colors. Let $\what{c}$ be this proper coloring of $G[W]$. The tuple $(W_i:= \{v\in W\mid \what{c}(v) = i\})_{i\in [k]}$ is called a {\em $k$-partial Grundy witness} for $G$. 
Now we apply the color-coding step of the algorithm. That is, we color the vertices of $G$ using $k$ colors independently and uniformly at random, and let $Z_1, \cdots, Z_k$ be the color classes of this coloring. The probability that for each $i\in [k]$, $W_i\subseteq Z_i$, is $k^{-k^2}$. Notice that since $G$ is a $d$-degenerate graph, we have that $G_i=G[Z_i]$, for each $i\in [k]$, is $d$-degenerate. Now we exploit this fact and apply the independence covering lemma of Lokshtanov et al.~\cite{lokshtanov2020covering}. That is given as input $(G_i,k^2)$, in time $2^{\OO(dk^2)}n^{\OO(1)}$ it produces a family
$\F_i$, of size $2^{\OO(dk^2)} \cdot \log n$ that contains independent sets of $G_i$. Furthermore, given any independent set $I$ of $G_i$ of size at most $k^2$, there exists an independent set $F\in \F_i$, such that $I\subseteq F$  ($F$ covers $I$). In particular, we know that there is a set $F_i\in \F_i$ that covers $W_i$. So, the algorithm just enumerates each tuple $(F_1,\ldots,F_k)$ of $\F_1\times \cdots \times \F_k$ and checks whether $(F_1,\cdots, F_k)$ is a $k$-partial Grundy coloring of $G[F_1\cup \cdots \cup F_k]$. If $(G,k)$ is a yes instance, our algorithm is successful with probability $k^{-k^2}$. Moreover, we can convert the described randomized algorithm to a deterministic one by using the standard derandomization technique of ``universal sets''~\cite{naor1995splitters,fahad2015dynamic}. Some remarks are in order; it can be shown that each of $|W_i|\leq k$, and hence we can call the independence covering lemma on $(G_i,k)$, resulting in an improved running time of $2^{\C{O}(d k^2 \log k)} \cdot n^{\C{O}(1)}$.  Aboulker et al.~\cite{aboulker2022grundy} proved that \pgc\ on $K_{t,t}$-free graphs is FPT, which includes $t$-degenerate graphs. Aboulker et al. did not explicitly mention the running time, but their running time is at least $2^{k^t}n$.  Our new algorithm improves over this.

For general instances, we do not have a bound on the degeneracy of the input graph.  However, we can achieve this by using our new structural result. 


\begin{restatable}[Structural result]{theorem}{StructResult}
\label{thm:deg_reduction}
There is a polynomial-time algorithm that given a graph $G$ and a positive integer $k$, does one of the following:
\begin{itemize}
\item[(i)] Correctly concludes that $(G,k)$ is a yes-instance of \pgc, or
\item[(ii)] Outputs at most $2k^3$ induced bicliques $A_1\cdots,A_{\ell}$
in $G$ such that the following holds. 
For any $v \in V(G)$, the degree of $v$ in $G-F$ is at most $k^3$, where $F$ is the union of the edges in the above bicliques.
\end{itemize}
\end{restatable}

The structural result (\Cref{thm:deg_reduction}) is one of our main technical contributions. Next, we show how to design an algorithm for \pgc using \Cref{thm:deg_reduction}. We follow the same steps as for the one described for the degenerate case. That is, we have color classes $Z_i$s and they contain the respective $W_i$s (the part of the small witness set $W$). Now, to design a family of independent sets in $G_i=G[Z_i]$, we do as follows. Let $(L_j,R_j)$ be a bipartition of $A_j$, for each $j\in[\ell]$. Observe that any independent set $I$ (in particular of $G_i$) intersects $L_j$ or $R_j$, but {\em not both}, for any $j \in[\ell]$. Thus, we first guess whether $W_i$ intersects $L_j$, $R_j$ or none. Let this be given by a function $f_i:[\ell]\rightarrow \{L, R, N\}$, that is, if $f_i(j)=L$, then $W_i \cap L_j=\emptyset$, if $f_i(j)=R$, then $W_i\cap R_j=\emptyset$, else $W_i\cap (L_j \cup R_j)=\emptyset$. Taking advantage of this property, for each guess of which of $L_j$ or $R_j$ is not contained in $W_i$, we delete the corresponding set (which is one of $L_j$ or $R_j$, for each $j\in [\ell]$) from $G_i$. We call the resulting graph $G_i^{f_i}$. This implies that for any $f_i$, in $G_i^{f_i}$ we delete all edges of $F$ (where $F$ is the union of edges in the bicliques). Hence, the maximum degree of $G_i^{f_i}$ is at most $k^3$, and therefore it has degeneracy at most $k^3$. Now using the independence covering step of the algorithm for degenerate graphs, we can finish the algorithm. The proof \Cref{thm:deg_reduction} is obtained by carefully analyzing the reason for the failure of a greedy algorithm.

Our next result is an affirmative answer to Question $2$. 

\begin{theorem}\label{thm:KijGrundy}
For any fixed $i,j \in \mathbb{N}$, \Grundycol\ on $K_{i,j}$-free graphs admits an FPT algorithm when parameterized by $k$.
\end{theorem}

For our algorithm, we use a reinterpretation of the problem which is based on the existence of a {\em small witness}. 
Gy{\'{a}}rf{\'{a}}s et al.~\cite{DBLP:journals/siamdm/GyarfasKL99}, and Zaker~\cite{zaker2006results} independently showed that a given instance $(G,k)$ of \Grundycol\ is a yes-instance if and only if there is a vertex subset $W$ of size at most $2^{k-1}$, such that $(G[W],k)$ is also a yes-instance of the problem. The existence of this small induced subgraph directly yields an XP algorithm for the problem~\cite{zaker2006results}. Using characterizations of ~\cite{DBLP:journals/siamdm/GyarfasKL99,zaker2006results} and basic Grundy coloring properties, we can reduce \Grundycol to finding a homomorphic image, satisfying some independence constraints, of some {\em specific labeled trees}. Let the pair
$(T,\ell)$ denote a rooted tree $T$ together with a labeling function $\ell:V(T)\rightarrow [k]$. Given $(T,\ell)$ and a graph $G$, a function $\omega: V(T) \rightarrow V(G)$ is a {\em labeled homomorphism} if: i) for each $\{u,v\}\in E(T)$, we have $\{\omega(u),\omega(v)\}\in E(G)$, and ii) for $u,v\in V(T)$, if $\ell(u)\neq \ell(v)$, then $\omega(u)\neq\omega(v)$. 
In particular, we will reduce the problem to the following.   

\smallskip

\defparproblem{\sc Constrained Label Tree Homomorhism (CLTH)}{A host graph $G$ and a labeled tree pattern $(T,\ell)$}{$|V(T)|$}{Does there exists a labelled homomorphism $\omega: V(T) \rightarrow V(G)$ such that for any $z \in[k]$, $W_i = \{\omega(t) \mid t\in V(T) \mbox{ and $\ell(t)=i$}\}$ is an independent set in $G$?}

\medskip

Now our goal is to identify $\omega$ (and thus the witness set $W$). The first step of our algorithm will be to use the color-coding technique of  
Alon-Yuster-Zwick.~\cite{DBLP:journals/jacm/AlonYZ95} to ensure that the labeling requirement of the graph homomorphism $\omega$ is satisfied. To this end, we randomly color the vertices of $G$ using $k$ colors, where we would like the random coloring to ensure that for each $z \in [k]$, all the vertices in $W_z$ are assigned the color $z$. Let $X_1,X_2,\cdots, X_k$ be the color classes in a coloring that achieves the above property. Our objective will be to find $\omega$ such that vertices of $T$ that are labeled $z\in[k]$ are assigned to vertices in $X_z$. 

Our next challenge is to find a homomorphism that additionally satisfies the independence condition. That is, vertices of the same label in $T$ are assigned to an independent set in $G$. Note that the number of potential $\omega$s that satisfy our requirements can be very huge; however, we will be able to carefully exploit $K_{i,j}$-freeness to design a dynamic programming-based algorithm to identify one such $\omega$ (and thus the set $W$). Our approach here is inspired by dynamic programming in the design of FPT algorithms based on computations of ``representative sets''~\cite{DBLP:journals/tcs/Marx09,DBLP:journals/jacm/FominLPS16}. However, this inspiration ends here, as to apply known methods we need to have an underlying family of sets that form a matroid. Unfortunately, we do not have any matroid structure to apply the known technique. Here, we exploit the fact that we have a specific labeled tree $(T,\ell)$ and a $K_{i,j}$-free graph. Next, we define the specific trees that we will be interested in [refer Figure \ref{fig:GrundySet-tree}(ii)].

\begin{figure}[t]
\centering
\includegraphics[scale=0.8]{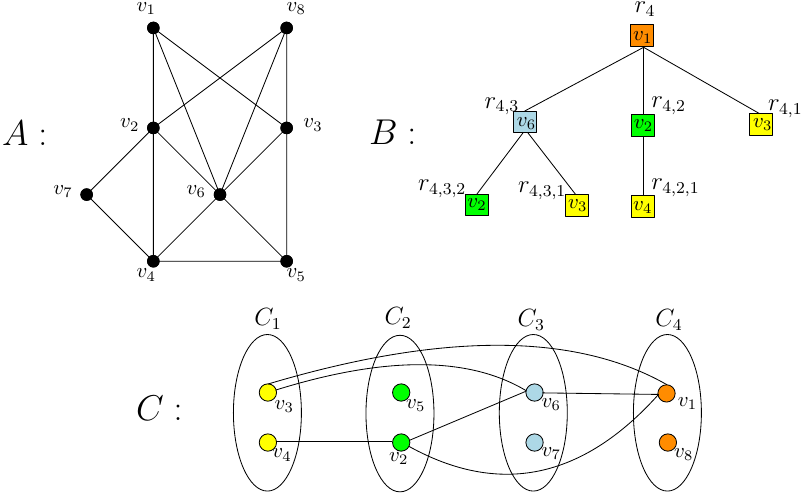}
\caption{Illustration of a labelled homomorphism $\omega : V (T_4) \rightarrow V (G)$, where the graph is
shown in part $(i)$, $T_4$ with $V (T_4) = \{r_4, r_{4,3}, r_{4,2}, r_{4,1}, r_{4,3,2}, r_{4,3,1}, r_{4,2,1}\}$ is shown in part $(ii)$,
and a relation to Grundy coloring is illustrated in part $(iii)$. Here, $\omega(r_4) = v_1, \omega(r_{4,3}) = v_6,
\omega(r_{4,2}) = v_2, \omega(r_{4,1}) = v_3, \omega(r_{4,3,2}) = v_2, \omega(r_{4,3,1}) = v_3$, and $\omega(r_{4,2,1}) = v_4$.}
\label{fig:GrundySet-tree}
\end{figure}

\begin{definition}\label{def:grundy-tree-new}{\rm
For each $k \in \mathbb{N}\setminus \{0\}$, we (recursively) define a pair $(T_k, \ell_k: V(T_k) \rightarrow [k])$, called a {\em $k$-Grundy tree}, where $T_k$ is a tree and $\ell_k$ is a {\em labeling} of $V(T_k)$, as follows\longversion{ (see Fig.~\ref{fig:Grundy-tree})}:
\begin{enumerate}
\item $T_1 = (\{r_1\}, \emptyset)$ is a tree with exactly one vertex $r_1$ (which is also its root), and $\ell_1(r_1) = 1$.

\item Consider any $k \geq 2$, we (recursively) obtain $T_{k}$ as follows. For each $z \in [k-1]$, let $(T_z,\ell_z)$ be the $z$-Grundy tree with root $r_z$. We assume that for distinct $z,z' \in [k-1]$, $T_z$ and $T_{z'}$ have no vertex in common, which we can ensure by renaming the vertices.\footnote{For the sake of notational simplicity we will not explicitly write the renaming of vertices used to ensure pairwise vertex disjointness of the trees. This convention will be followed in the relevant section.} We set $V(T_k) = \big(\cup_{z \in [k-1]}V(T_{z})\big) \cup \{r_k\}$ and $A(T_k) = \big(\cup_{z \in [k-1]}A(T_{z})\big) \cup \{(r_k, r_{z}) \mid z \in [k-1]\}$. We set $\ell_k(r_k) = k$, and for each $z \in [k-1]$ and $t \in V(T_z)$, we set $\ell_k(t) = \ell_z(t)$.
\end{enumerate}
For $v \in V(T_k)$, $\ell_k(v)$ is the {\em label} of $v$ in $(T_k, \ell_k)$, and the elements in $[k]$ are {\em labels} of $T_k$.
}\end{definition}

Observe that the label of a vertex $t\in V(T)$ is the depth of the subtree rooted at $t$ (the depth of a leaf is $1$). In particular, the leaves are assigned the label $1$, and when they are deleted, we get vertices with the label $2$ as leaves, and so on.  This allows us to perform a bottom-up dynamic programming over $T_k$. Roughly speaking, for each $z\in [k]$, we solve the special labelled tree homomorphism from $\omega_z: V(T_z) \rightarrow X_1\cup X_2 \ldots, X_z$, where the root of $T_z$ is mapped to a fixed vertex $v\in X_z$ as follows: instead of having all potential choices for $\omega_z$ (or $W_z =\{\omega_z(t) \mid t\in V(T_z)\}$), we find enough representatives, that will allow us to replace $W_z$ by something that we have stored. It is a priori not clear that such representative sets of small size exist and furthermore, even if they exist, how to find them. The existence and computation of small representative sets in this setting is our main technical contribution to \Grundycol.



We heavily exploit the $K_{i,j}$-freeness in our ``representative set'' computation. Very roughly stating, while we have computed required representatives for $W_z$, and wish to build such a representatives for $W_{z+1}$, by exploiting $K_{i,j}$-freeness, we either find a small hitting set or a large sub-family of pair-wise disjoint sets. In the former case, we can split the family and focus on the subfamily containing a particular vertex from the hitting set and obtain a ``representative'' for it (and then take the union over such families). In the latter case, we show that we are very close to satisfying the required property, except for the sets containing vertices from an appropriately constructed small set $S$ of vertices. The construction of this {\em small} set $S$ is crucially based on the $K_{i,j}$-freeness of the input graph. Once we have the set $S$, we can focus on sets containing a vertex from it and compute ``representatives'' for them.

Again, using standard hash functions, we can obtain a deterministic FPT algorithm for the problem by derandomizing the color coding based step~\cite{DBLP:journals/jacm/AlonYZ95,naor1995splitters}.

\section{Preliminaries}

\noindent{\em Generic Notations.} We denote the set of natural numbers by $\mathbb{N}$. For $n\in \mathbb{N}$, $[n]$ denotes the set $\{1,2,\cdots,n\}$. For a function $f: X \rightarrow Y$ and $y\in Y$, $f^{-1}(y) := \{x\in X \mid f(x) = y\}$. 

For standard graph notations not explicitly stated here, we refer to the textbook of~Diestel~\cite{diestel2005graph}. For a graph $G$, we denote its vertex and edge set by $V(G)$ and $E(G)$, respectively. 
 The neighborhood of a vertex $v$ in a graph $G$ is the set of vertices that are adjacent to $v$ in $G$, and we denote it by $N_G(v)$. The degree of a vertex $v$ is the size of its neighborhood in $G$, and we denote it by $d_G(v)$. For a set of vertices $S \subseteq V (G)$, we define $N_G (S) = (\cup_{v\in S} N(v)) \setminus S$. When the graph is clear from the context, we drop the subscript $G$ from the above notations. For $X \subseteq V(G)$, the induced subgraph of $G$ on $X$, denoted by $G[X]$, is the graph with vertex set $X$ and edge set $\{\{u,v\}  \mid u,v \in X\  \&\ \{u,v\}\in E(G)\}$. Also, $G[V(G)\setminus X]$ is denoted by $G-X$. For $v \in V (G)$, we use
    $G - v$ to denote $G - \{v\}$ for ease of notation. For an edge subset $F\subseteq E(G)$, $G-F$ is the graph with vertex set $V(G)$ and edge set $E(G)\setminus F$. A bipartite graph $G=(A \uplus B,E)$ is called a \emph{biclique} if every vertex in $A$ is adjacent to every vertex in $B$. We assume that $A$ and $B$ are both non-empty sets. 
    For $d\in \mathbb{N}$, a graph is {\em $d$-degenerate} if each of its subgraphs has a vertex of degree at most $d$. 
    For terminologies related to parameterized complexity, we refer to the textbook of Cygan et al.~\cite{cygan2015parameterized}.

\section{FPT Algorithm for Partial Grundy Coloring}

Consider a graph $G$ and an integer $k\in \mathbb{N}\setminus \{0\}$.~For a (not necessarily proper) coloring $\chi: V(G)\rightarrow [k]$, for simplicity, we sometimes write $\chi$ as the ordered tuple $(\chi^{-1}(1),\chi^{-1}(2),\cdots,$ $\chi^{-1}(k))$. Recall that a proper coloring of a graph is a coloring of its vertices so that for none of its edges, the two endpoints of it are of the same color.  Also, a $k$-partial Grundy coloring of $G$ is a proper coloring $c: V(G)\rightarrow [k]$, such that for each $i\in [k]$, there is a vertex $v\in V(G)$ with: $(i)$ $c(v) =i$ and $(ii)$ for every $j \in [i-1]$, there is $u\in N_G(v)$ with $c(u) = j$.

We will begin with some definitions and results that will be useful in obtaining our main structural result (\Cref{thm:deg_reduction}) and our FPT-algorithm (\Cref{thm:pgc}). 

Note that given a $k$-partial Grundy coloring of an induced subgraph $\what{G}$ of a graph $G$, we can extend this coloring to a partial Grundy coloring of the whole graph $G$ using at least $k$ colors by greedily coloring the uncolored vertices of $G-V(\what{G})$.
The following observation will be particularly useful when we work with a small ``witness''.

\longversion{\begin{observation}}
\shortversion{\begin{observation}[$\spadesuit$]\footnote{Proofs of results marked with $\spadesuit$ can be found in the appended full version of the paper.}}
\label{prop:extension}
Given a graph $G$, an induced subgraph $\what{G}$ of $G$, and a partial Grundy coloring of $\what{G}$ using $k$ colors, we can find a partial Grundy coloring of $G$ using at least $k$ colors in linear time.
\end{observation}

\longversion{
\begin{proof}
Let $\what{c}: V(\what{G}) \rightarrow [k]$ be a partial Grundy coloring of $\what{G}$ with exactly $k$ colors. We construct a partial Grundy coloring $c: V(G) \rightarrow \mathbb{N}$ of $G$ using at least $k$ colors as follows. For each vertex $v \in V(\what{G})$, set $c(v) := \what{c}(v)$. Let $v_1,v_2, \cdots, v_{n'}$ be an arbitrarily fixed order of vertices in $V(G) \setminus V(\what{G})$. Let $G_0 =\what{G}$, and for each $p \in [n']$, $G_p = G[V(\what{G} \cup \{v_1,v_2, \cdots, v_p\})]$. We iteratively create a partial Grundy coloring $c_p$ of $G_p$ using at least $k$ colors (in increasing values of $p$) as follows. Note that $c_0=\what{c}$ is already a partial Grundy coloring of $G_0$ that uses at least $k$ colors. Consider $p \in [n'] \setminus \{0\}$, and assume that we have already computed a partial Grundy coloring $c_{p-1}: V(G_{p-1}) \rightarrow \mathbb{N}$ of $G_{p-1}$ that uses $k' \geq k$ colors. For each $z \in [k']$, let $V_z = c_{p-1}^{-1}(z)$. For each $v \in V(G_{p-1})$, we set $c_p(v) := c_{p-1}(v)$. If the vertex $v_p$ has a neighbor in each of the sets $V_1,V_2, \cdots, V_{k'}$, i.e., if for each $z \in [k']$, $N_G(v_p) \cap V_z \neq \emptyset$, then set $c_p(v_p) := k'+1$. Otherwise, let $z^* \in [k']$ be the smallest number such that $N_G(v_p) \cap V_{z^*} = \emptyset$, and set $c_p(v_p) := z^*$. Notice that by construction, $c_p$ is a partial Grundy coloring of $G_p$ using at least $k$ colors. From the above discussions, $c_{n'}$ is a partial Grundy coloring of $G=G_{n'}$ using at least $k$ colors.
 \end{proof}
}



We now define a ``witness'' for a partial Grundy coloring.

\begin{definition}\label{def:grundy-witness}{\rm
Consider a graph $G$ and an integer $k \in \mathbb{N}\setminus \{0\}$. A sequence of pairwise disjoint independent sets $(Q_1,Q_2,\ldots,Q_k)$ of $G$ is a \emph{$k$-partial Grundy witness} if the following holds. For any $i \in [k]$, there is $v \in Q_i$ such that for all $j\in [i-1]$, $Q_j\cap N_G(v)\neq \emptyset$. The vertex $v$ is called the {\em dominator} in $Q_i$. 
}\end{definition}

The next observation follows directly from the above definition.

\begin{observation}\label{obs:k-witness-supset}
Given a graph $G$ and an integer $k$, let $(Q_1,Q_2,\ldots,Q_k)$ be a $k$-partial Grundy witness.
Suppose $Y_1,Y_2,\ldots,Y_k$ are pairwise disjoint independent sets in $G$ such that $Q_i\subseteq Y_i$, for all $i\in [k]$. Then $(Y_1,Y_2,\ldots,Y_k)$ is also a $k$-partial Grundy witness of $G$.
\end{observation}
A $k$-partial Grundy witness $(X_1,X_2,\ldots,X_k)$ is {\em small} if for each $i\in [k]$, $|X_i|\leq k-i+1$. Next, we prove the existence of a small $k$-partial Grundy witness. This result is the same as the one obtained by Effantin et al.~\cite{effantin2016characterization}; however, it is stated slightly differently for convenience.

\longversion{\begin{observation}}
\shortversion{\begin{observation}[$\spadesuit$]}
\label{obs:k-witness}
Given a graph $G$ and an integer $k$, where $G$ has a $k$-partial Grundy witness. Then, there exists a $k$-partial Grundy witness $(X_1,X_2,\ldots,X_k)$ such that for each $i\in [k]$, $X_i\subseteq Q_i$ and $|X_i|\leq k-i+1$.
\end{observation}
\longversion{
\begin{proof}
Let $(Q_1,Q_2,\ldots,Q_k)$ be a $k$-partial Grundy witness in $G$.
For each $i\in [k]$, arbitrarily choose a vertex $u_i \in Q_i$ such that $N_G(u_i)\cap Q_j \neq \emptyset$ for all $j\in [i-1]$. We mark $u_i$ and one vertex from $N_G(u_i)\cap Q_j$ for all $j\in [i-1]$. Let $M$ be the set of all marked vertices, and for $i \in [k]$, let $X_i=M\cap Q_i$. Then, $(X_1,X_2,\ldots,X_k)$ is the required $k$-partial Grundy witness.
\end{proof}
}

The remainder of this section is organized as follows. In Section~\ref{ingr:one} we prove our key structural result (\Cref{thm:deg_reduction}), and then obtain our algorithm in Section~\ref{sec:main-pgc}. (Readers who may want to read the algorithm directly, may skip Section~\ref{ingr:one}.)

\subsection{Degree Reduction: Proof of Theorem~\ref{thm:deg_reduction}}\label{ingr:one}

%

One of the key ingredients of our algorithm is the following structural theorem.

\StructResult*

\smallskip

The objective of this section is to prove~\Cref{thm:deg_reduction}. The proof of this theorem is based on the following lemma proved for bipartite graphs.

\begin{lemma}
\label{lem:twoside}
There is a polynomial-time algorithm that, given a bipartite graph $G=(A\uplus B,E)$ and a positive integer $k$, does one of the following.
\begin{itemize}
\item[(i)] Correctly conclude that the partial Grundy coloring of $G$ is at least $k$.
\item[(ii)] Outputs at most $4k-4$ bicliques $A_1\cdots,A_{\ell}$ in $G$ such that for any $v \in V(G)$, degree of $v$ in $G-F$ is at most $k^2$, where $F$ is the union of the edges in the above bicliques.
\end{itemize}
\end{lemma}

We first give a proof of Theorem~\ref{thm:deg_reduction} based on the above lemma. 

\begin{proof}[Proof of Theorem~\ref{thm:deg_reduction}]
Consider a graph $G$ and a positive integer $k$. First, we run the first-fit greedy algorithm for proper coloring of the graph $G$ for an arbitrarily fixed ordering $(v_1,v_2,\cdots,v_n)$ of $V(G)$.
For each $j$, let $C_j$ be the vertices colored $j$ and $k'$ be the largest integer such that $C_{k'}\neq \emptyset$. Note that $(C_1,\cdots,C_{k'})$ is a proper coloring of $G$. Also, for any $j\in [k']$ and any vertex $v$ in $C_j$, $v$ has a neighbor in $C_{j'}$ for all $j'\in [j-1]$. If $k'\geq k$, then $(C_1,\cdots,C_{k'})$ is a partial Grundy coloring of $G$ using at least $k$ colors, and thus we can correctly report it.

Next, we assume that $k'<k$. Note that all the edges in $G$ are between the color classes $C_1,\cdots,C_{k'}$. Now for every distinct $i,j\in [k']$, where $i<j$, we apply Lemma~\ref{lem:twoside} on $(H_{i,j}=(C_i\uplus C_j, E(C_i,C_j)),k)$, where $E(C_i,C_j)$ is the set of edges in $G$ between the color classes $C_i$ and $C_j$. Our algorithm will declare that $G$ has a partial Grundy coloring using at least $k$ colors if we get the output given in statement $(i)$ in any of the $\binom{k'}{2}$ applications of Lemma~\ref{lem:twoside}. Otherwise, for every distinct $i,j\in [k']$, where $i<j$, let $A_{i,j,1},\cdots, A_{i,j,\ell_{i,j}}$ be the bicliques, we get as output by the algorithm in Lemma~\ref{lem:twoside} on $(H_{i,j}$ $=(C_i\uplus C_j, E(C_i,C_j)),k)$. Note that
$\ell_{i,j}\leq 4k-4$. Now our algorithm will output the bicliques $\{A_{i,j,r}~\colon~1\leq i<j\leq k', r \in [\ell_{i,j}]\}$. As $k'< k$ for all $1\leq i<j\leq k'$, the number of bicliques we output is at most $ (4k-4)\binom{k'}{2}$, that is, at most $ 2k^3$. Since any vertex in a color class $C_r$ has neighbors in other color classes, for $r\in [k']$ and we applied Lemma~\ref{lem:twoside} for every pair of color classes, the degree of $v$ in $G-F$ is at most $k^3$ for any $v \in V(G)$, where $F$ is the union of the edges in the above bicliques.
This completes the proof of the theorem.
\end{proof}

Next, we focus on the proof of \Cref{lem:twoside}. Toward this, we give a polynomial time procedure that, given a bipartite graph $G=(L\uplus R, E)$ and a positive integer $k$, either concludes that the input graph has partial Grundy coloring using at least $k$ colors or it outputs at most $2k-2$ bicliques $A_1\cdots,A_{\ell}$ in $G$ such that for any $v \in L$, $d_{G-F}(v)\leq k^2$, where $F$ is the union of the edges in the above bicliques. That is, removal of the edges of these bicliques bounds the degree of each vertex in $L$ by $k^2$.  We get the proof of \Cref{lem:twoside} by applying this algorithm once for $L$ and then for $R$. 

\vspace{2mm}

\noindent 
\textbf{Overview of our algorithm.}
 Let $\sigma=v_1,v_2,\ldots,v_{n}$ be an ordering of the vertices in $L$ in the non-increasing order of their degree in $G$.  
 The algorithm constructs {\em specific} color classes $Q_1,Q_{2},\ldots, Q_{r}$ in this order so that $(C_1=Q_r,C_{2}=Q_{r-1},\ldots, C_{r}=Q_1)$ is an $r$-partial Grundy witness, where $|Q_j|\leq j$. Furthermore, we will construct sets $B_i$, $i\in[r]$, which will be used to construct the bicliques. Notice that if we obtain $r\geq k$, then we will be able to conclude that $G$ has a partial Grundy coloring using at least $k$ colors. 
 Let $Q_{1}=\{v_1\}$ and in our construction $v_1$ will be the dominator in $Q_1$ (and our construction needs to ensure this property; see Definition~\ref{def:grundy-witness}). 
Consider the construction of $Q_{2}$. Let $i$ be the smallest index in $\{2,3,\ldots,n\}$ for which there is a vertex $w\in N_{G}(v_1)$ such that $v_i$ is not adjacent to $w$. Then, we set $Q_{2}=\{v_i,w\}$, and designate $v_i$ as the dominator in color class $C_{r-1} = Q_{2}$. Notice that all the vertices in $B_1=\{v_2,\ldots,v_{i-1}\}$ is adjacent to all the vertices in $N_{G}(v_1)$, and hence they together form a biclique (with bipartition $B_1$ and $N_{G}(v_1)$).
This property will be extended in building each $Q_j$s and the required bicliques. 

For the construction of $Q_j$, we consider {\em unprocessed vertices} (i.e., the vertices that do not belong to the previously constructed sets, i.e., to $Q_1\cup B_1\cup \ldots \cup Q_{j-1}\cup B_{j-1}$) as follows. We would now like to choose an unprocessed vertex $v_{i'}$, so that we can make $v_{i'}$ the dominator of $Q_{j}$, and additionally, for each $j'\in [j-1]$, we can include a neighbor of the dominator from $Q_{j'}$ to the set $Q_j$. Note for us to do the above, we need to ensure that the vertices that we add to $Q_{j}$ is an independent set in $G$, and all the vertices that we want to include in the set $Q_{j}$ are outside $Q_1\cup \ldots \cup Q_{j-1}$. That is, among the unprocessed vertices, we choose the first vertex $v_{i'}$ with the following property: for each $j'\in [j-1]$, we have a neighbor $w_{j'}$ of the previously constructed dominator in $Q_{j'}$ 
such that $w_{j'}\notin Q_1\cup \ldots \cup Q_{j-1}$ and $(v_{i'},w_{j'})\notin E(G)$; we set 
$Q_j=\{v_{i'},w_1,\ldots,w_{j-1}\}$. We would like to mention that all the dominators we construct are from the bipartition $L$ and hence $\{w_1,\ldots,w_{j-1}\}\subseteq R$. This will imply that $Q_j$ is an independent set.  



Moreover, by the choice of $v_{i'}$ as the smallest vertex with the desired property, it follows that for any vertex $v_r$ that appears before $v_{i'}$ in the order $\sigma$ and $v_r\notin P=Q_1\cup \ldots \cup Q_{j-1}$, the vertex $v_{r}$ is adjacent to all the vertices in $N(x_{j'}) \setminus P$, where $x_{j'}$ is the dominator in $Q_{j'}$, for some $j'\in [j-1]$. Then we add $v_r$ to $B_{j'}$. Note that in the above process, we still maintain the biclique property, by explicitly ensuing that $B_{j'}$ and $N(x_{j'})\setminus (Q_1\cup \ldots \cup Q_{j-1}\cup Q_j)$ forms a biclique. 

\begin{algorithm}[!ht]
\DontPrintSemicolon

Initialize $B_i:=\emptyset$ and $Q_i:=\emptyset$, for each $i\in [k]$.\;
Let $\sigma=v_1,\cdots,v_n$ be an order of the vertices in $L$ in the non-increasing order of their degrees.\;
Let $x_1:=v_1$, $Q_1:=\{x_1\}$ and $q:=2$.\; \label{step3}
\While{$L \nsubseteq \bigcup_{j\in [q-1]} (Q_j\cup B_j)$}
{
Let $v_r$ be the first vertex in $\sigma$ from $L \setminus \big(\bigcup_{j\in [q-1]} (Q_j\cup B_j) \big)$.\; \label{step5}
Let $P = \cup_{j \in [q-1]} Q_{j}$.\; \label{step6}
\eIf{there exists $j\in [q-1]$ such that $N(x_j) \subseteq P \cup N(v_r) $}  
{
choose an arbitrary $j$ with this property and set $B_j:=B_j \cup \{v_r\}$ \longversion{(see Fig.~\ref{fig:algo-PGC-oneside}}). \; \label{step8}
}
{
For each $j\in [q-1]$, $N(x_j)\setminus (P\cup N(v_r))\neq \emptyset$. Then, for each $j\in [q-1]$ arbitrarily pick a vertex $w_j$ from $N(x_j)\setminus (P\cup N(v_r))$. (Notice that $w_j$ {\em maybe equal to} $w_{j'}$ for two distinct $j,j'\in [q-1]$). \; \label{step10}
Set $x_q:=v_r$\; \label{step11}
Set $Q_q:=\{x_q\}\cup \{w_1,\cdots,w_{q-1}\}$\; \label{step12}
Set $q:=q+1$\; \label{step13}
}

}
\eIf{$q\geq k+1$}{
Declare that partial Grundy coloring of $G$ is at least $k$.\label{stepyes}
}
{
For all $j\in [q-1]$, let $A_j$ be the bipartite graph induced on $B_j$ union $N(x_j)\setminus P$ and $S_j$ be the graph induced on $N[x_j]$, where $P=\bigcup_{j\in [q-1]} Q_j$.\; \label{stepno1}
Output $A_1,\cdots,A_{q-1}$ and $S_1,\cdots, S_{q-1} $ \label{stepno2}
}

\caption{Algorithm for one-sided bipartite structural result}
\label{alg:bipartitestruc}
\end{algorithm}

\noindent 
\textbf{Description of the algorithm:} First, we intialize the sets $B_i$ and $Q_i$ to be the empty set, for all $i\in [k]$ (see Algorithm~\ref{alg:bipartitestruc}). Let $\sigma=v_1,v_2,\ldots,v_n$ be an ordering of the vertex set $L$ in the non-increasing order of their degrees. Now, we want to construct the color classes $Q_1,Q_2,\ldots,Q_k$, iteratively, such that $(C_1,C_2,\ldots,C_k)=(Q_k,Q_{k-1},\ldots,Q_1)$ is a $k$-partial Grundy witness. 
At line \ref{step3}, we intialize with $Q_1:=\{v_1\}$, 
$x_1:=v_1$, and fix the index $q=2$.  Here,  $Q_1$ will be the color class $C_k$ with dominator vertex $x_1$. 
Now consider an iteration of the {\bf  while} loop. 
The algorithm checks if the set $L \setminus \big(\bigcup_{j\in [q-1]} (Q_j\cup B_j) \big)$ is non-empty and executes the {\bf while} loop.  At this point we have constructed sets $Q_1,\ldots,Q_{q-1}$ such that $(Q_{q-1},Q_{q-2},\ldots,Q_1)$ is a $(q-1)$-partial Grundy witness such that each $x_i$ is a dominator vertex in $Q_i$. 
Let $v_r$ be the first unprocessed vertex in $L$ and $P = \cup_{j \in [q-1]} Q_{j}$ by Lines \ref{step5} and \ref{step6}. Now, we check if we can construct the current color class $Q_q$ with vertex $v_r$ as a dominator vertex, and for that, we need to add a neighbor $w_j$ (which is not added to any $Q_{i'}$ before) for each already discovered dominator $x_j$ such that $w_j$ is non-adjacent to $v_r$.  
Now, if there exists some $j\in [q-1]$ such that each neighbor of $x_j$ is a neighbor of $v_r$ (see Fig.~\ref{fig:algo-PGC-oneside}(i)),  then we will not be able to construct $Q_q$ with vertex $v_r$ in it. 
In that case,  we choose such a value $j$ and add $v_r$ to $B_j$(See Line~\ref{step8}).  Here,  notice that  $v_r$ is adjacent to all the vertices $N({x_j})\setminus P$.  We will maintain this property to all the vertices added to $B_j$,  i.e.,  $B_j$ union $N({x_j})\setminus \bigcup_i Q_i$ forms a biclique. 
Now consider the case that the condition in the {\bf if} statement in Line~7 is false. 
Then, we choose a vertex $w_j\in N(x_j)\setminus (P\cup N(v_r))$ for each $j\in [q-1]$, by line \ref{step10} and set the vertex $v_r$ as the dominator for $Q_q$, that is, $Q_q:=\{x_q\}\cup \{w_1,\cdots,w_{q-1}\}$,  at Line \ref{step12}. 
Notice that $Q_q$ is an independent set because there is no edge between $x_1$ and a vertex in $\{w_1,\ldots,w_{q-1}\}$, and $\{w_1,\ldots,w_{q-1}\}$ is a subset of $R$, the right part of the bipartition of $G$. 
We repeat the iteration until one of the {\bf while} loop conditions at line 4 fails.  Next, if $q\geq k+1$, we conclude that $G$ has a partial Grundy coloring using $k$ colors by line \ref{stepyes}, because $(Q_{q-1},\ldots,Q_1)$ is a $(q-1)$-partial Grundy witness, where $q-1\geq k$.  Otherwise,  by line \ref{stepno1}, let $A_j$ be the graph induced on $B_j\cup (N(x_j)\setminus P)$ and $S_j$ be the graph induced on $N[x_j]$, for each $j\in [q-1]$.  Recall that $A_j$ is a biclique.  
It is easy to see that $S_j$ is a biclique,  because $G$ is a bipartite graph. 
At line \ref{stepno2}, the algorithm outputs the set of graphs $A_1,\cdots,A_{q-1}$ and $S_1,\cdots, S_{q-1} $.

Since the number of iterations of the {\bf while} loop is at most $n$ and each step in the algorithm takes polynomial time, the total running time of the algorithm is polynomial in the input size.
Next, we prove the correctness of the algorithm.

\begin{lemma}

Algorithm~\ref{alg:bipartitestruc} is correct.
\end{lemma}

\begin{proof}

Let $q^{\star}$ be the value of $q$ at the end of the algorithm. To prove the correctness of the algorithm, first, we prove the following claim.

\longversion{
\begin{figure}
\centering
\includegraphics[scale=0.8]{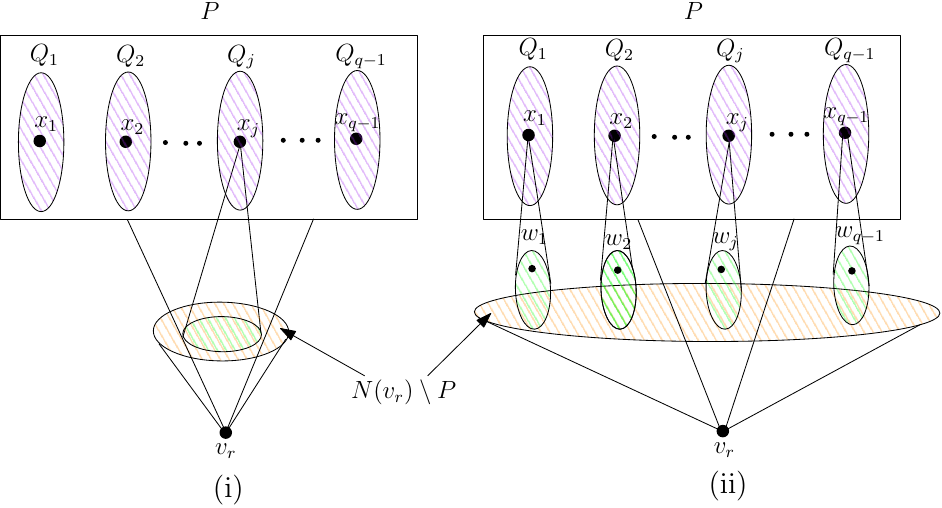}
\caption{An illustration for the cases at Line~7 (on the left) and Line \ref{step10} (on the right) of Algorithm~1. For the vertex $x_j$, green shaded region represents the $N(x_j)\setminus P$.}
\label{fig:algo-PGC-oneside}
\end{figure}
}

\shortversion{\begin{claim}[$\spadesuit$]}
\longversion{\begin{claim}}
\label{clm:corr}
The following statements are true.
\begin{itemize}
\item[(i)] For each $i\in [q^{\star}-1]$, $Q_i$ is an independent set and $Q_i\neq \emptyset$
\item[(ii)] For each $i\in [q^{\star}-1]$ and $j\in [i-1]$, $N(x_j)\cap Q_i \neq \emptyset$.
\item[(iii)] For each $i\in [q^{\star}-1]$ and $v\in B_i$, $v$ is adjacent to all the vertices in $N(x_i) \setminus (\bigcup_{j\in [q^{\star}-1]} Q_j)$ and $d_G(v)\leq d_G(x_i)$.
\end{itemize}
\end{claim}
\longversion{
\begin{proof}
We prove the statements by induction on $i$. The base case is when $i=1$. Clearly, $Q_1=\{x_1\}$ and hence statement (i) is true. Statement (ii) is vacuously true. Next, we prove statement (iii). Notice that in any iteration of the {\bf while} loop, in Step~\ref{step8}, we may add a vertex $v_r$ to $B_1$. If this happens, then we know that $N(x_1)\subseteq P\cup N(v_r)$, where $P$ is a subset of $\bigcup_{j\in [q^{\star}-1]} Q_j$. That is, all the vertices in $N(x_1)\setminus P$ are adjacent to $v_r$. This implies that $v_r$ is adjacent to all the vertices in $N(x_1) \setminus (\bigcup_{j\in [q^{\star}-1]} Q_j)$. Since $x_1$ is the vertex with the maximum degree, we have that
$d_G(v_r)\leq d(x_1)$.

Next, for the induction step, we assume that the induction hypothesis is true for $i-1$, and we will prove that the hypothesis is true for $i$. Consider the iteration $h^{\star}$ of the {\bf while} loop when $q=i$ and Steps~\ref{step10}-\ref{step13} is executed. Let $v_r$ be the vertex mentioned in Step~\ref{step5} during that iteration. The vertices $w_1,\cdots,w_{q-1}$ belongs to $R$ (the right side of the bipartition of $G$) and hence $\{w_1,\cdots,w_{q-1}\}$ is an independent set. Also, notice that each $w_j$ does not belong to $N(v_r)$ (See Step~\ref{step10}). Hence, $Q_q:=\{v_r\}\cup \{w_1,\cdots,w_{q-1}\}$ is an independent set and $Q_q\neq \emptyset$. Thus, we proved statement (i).
Again, notice that $w_j \in N(x_j)$ for all $j\in [q-1]$ (See Step~\ref{step10}). Thus, statement (ii) follows.
Next, we prove statement (iii), which is similar to the proof of it in the base case.
Notice that in any iteration of the while loop (after the iteration $h^{\star}$), in Step~\ref{step8}, we may add a vertex $v_{r'}$ to $B_i$. If this happens, then we know that $N(x_i)\subseteq P\cup N(v_{r'})$, where $P$ is a subset of $\bigcup_{j\in [q-1]} Q_j$. That is, all the vertices in $N(x_i)\setminus P$ are adjacent to $v_{r'}$. This implies that $v_{r'}$ is adjacent to all the vertices in $N(x_i) \setminus (\bigcup_{j\in [q-1]} Q_j)$. Since $x_i \in Q_i$, considered in iteration $h^{\star}$, $d_G(v_{r'})\leq d_G(x_i)$ (See Step~\ref{step5}).
This completes the proof of the claim.
\end{proof}
}
\longversion{
\begin{figure}[h!]
    \centering
    \includegraphics{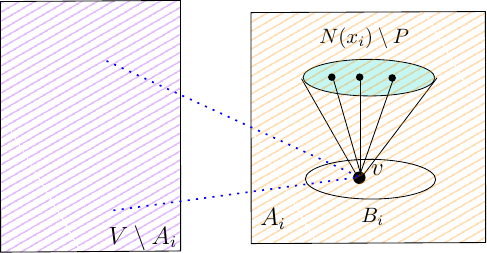}
    \caption{Here the vertex $v\in B_i$ in the biclique $A_i$ (right side). The number of neighbors of $v$ outside $A_i$ (blue edges) cannot be more than $|P|$ as $d_G(v)\leq d_G(x_i)=|N_G(x_i)|$ and $v$ has $|N(x)\setminus P|$ neighbours in the biclique $A_i$. }
    \label{fig:deg-bound}
\end{figure}}

Now suppose $q^{\star}\geq k+1$. Then, by Statements (i) and (ii) in Claim~\ref{clm:corr}, we get that $(Q_k,\cdots, Q_1)$ is a partial Grundy coloring of the graph induced on $\bigcup_{j\in [q^{\star}-1]} Q_j$. Thus, if the algorithm executes Step~\ref{stepyes}, then it is correct because of~\Cref{prop:extension}.

Now suppose $q^{\star}\leq k$. Then the algorithm executes Step~\ref{stepno2} and outputs the sets $A_1\cdots, A_{q^{\star}-1}$ and $S_1,\cdots,S_{q^{\star}-1}$. 
Statement (iii) in Claim~\ref{clm:corr} implies that each $A_j$ is a biclique in $G$, where $j\in [q^{\star}-1]$. Also, note that $S_j$ is a biclique in $G$ as it is induced on the set $N[x_j]$, for each $j\in [q^{\star}-1]$. Let $F$ be the union of the edges in the
bicliques $A_1,\cdots,A_{q^{\star}-1}$ and $S_1,\cdots,S_{q^{\star}-1}$.
Next, we prove that for any $v \in L $, $d_{G-F}(v)\leq k^2$. 
Let $X=\{x_1,\ldots,x_{q^{\star}-1}\}$.
It is easy to see that $|Q_j|\leq k$ and $L \cap Q_j=\{x_j\} \subseteq X$, for all $j\in [q^{\star}-1]$ (See Steps~\ref{step10}-\ref{step12}). Hence, $|\bigcup_{j\in [q^{\star}]} Q_j| \leq k^2$. Let $v$ be an arbitrary vertex in $L$. Note that if $v\in X$, the degree of $v$ in $G-F$ is zero (by the definition of $S_j$). Next, suppose that $v\in L\setminus X$.
Since $L \subseteq \bigcup_{j\in [q^{\star}-1]} (Q_j\cup B_j)$ (this is the condition for {\bf while} loop to exit) and $L \cap Q_{j} \subseteq X$ for all $j\in [q^{\star}-1]$, $v$ belongs to $B_i$ for some $i\in [q^{\star}-1]$. By statement (iii) in Claim~\ref{clm:corr}, $v$ is adjacent to all the vertices in $N(x_i) \setminus (\bigcup_{j\in [q^{\star}-1]} Q_j)$ and $d_G(v)\leq d_G(x_i)$. 
Recall that the biclique $A_i$ is the graph with bipartition $B_i$ and $N(x_i) \setminus (\bigcup_{j\in [q^{\star}-1]} Q_j)$. Moreover $v\in B_i$ and $d_G(v)\leq d_G(x_i)=|N_G(x_i)|$. 
This implies that 
the number of neighbors of $v$ that does not belong to 
$A_i$ is at most $|\bigcup_{j\in [q^{\star}-1]} Q_j|$, which is upper bounded by 
$k^2$. 
Therefore, the degree of $v$ in $G-F$ is at most $k^2$. See \Cref{fig:deg-bound} for an illustration. This concludes the proof.
\end{proof}

Now we are ready to give the proof of Lemma~\ref{lem:twoside} by applying Algorithm~\ref{alg:bipartitestruc} twice.


\longversion{
\begin{proof}[Proof of Lemma~\ref{lem:twoside}]
The proof follows by applying Algorithm~~\ref{alg:bipartitestruc} twice, one with $L=A$ and the other with $L=B$.
First, we apply Algorithm~~\ref{alg:bipartitestruc} where $L=A$ and $R=B$. If the algorithm declares that $G$ has a partial Grundy coloring with at least $k$ colors, then we are done. Otherwise, it outputs the bicliques
$A'_1\cdots,A'_{\ell_1}$ in $G$, where ${\ell}_1\leq 2k-2$ such that for any $v \in A $, the degree of $v$ in $G-F_1$ is at most $k^2$, where $F_1$ is the union of the edges in the above bicliques.

Next, again we apply Algorithm~\ref{alg:bipartitestruc} where $L=B$ and $R=A$. If the algorithm declares that $G$ has a partial Grundy coloring with at least $k$ colors, then we are done. Otherwise, it outputs another set of bicliques
$A''_1\cdots,A''_{\ell_2}$ in $G$, where ${\ell}_2\leq 2k-2$ such that for any $v \in B $, the degree of $v$ in $G-F_2$ is at most $k^2$, where $F_2$ is the union of the edges in the above bicliques.

Thus, we conclude that there is a partial Grundy $k$-coloring if any of the applications of Algorithm~\ref{alg:bipartitestruc} outputs a partial Grundy $k$-witness. Otherwise, we get at most $4k-4$ bicliques $A'_1,\cdots,A'_{\ell_1},A''_{1},\cdots,A''_{\ell_2}$ with the required property.
\end{proof}
}



\subsection{FPT Algorithm for Partial Grundy Coloring}\label{sec:main-pgc}

In this section, we design an FPT algorithm ${\cal A}$ for \pgc\ when the input has some structures. Then, we explain how to get an FPT algorithm for \pgc\ on general graphs using ${\cal A}$ and \Cref{thm:deg_reduction}. Moreover, our algorithm ${\cal A}$ will provide a faster algorithm for \pgc\ on $d$-degenerate graphs, improving the result in \cite{aboulker2022grundy}. Toward the first task, we define the following problem.  

\smallskip

\defproblem{\pgcstruct (\pgcstructshort)}{
Positive integers $k,d,\ell\in {\mathbb N}$, a  graph $G$, and $\ell$ bicliques $A_1\cdots,A_{\ell}$ in $G$ such that $G-F$ is $d$-degenerate, where $F$ is the union of the edges in the above bicliques.
}
{
Decide if there is a partial Grundy coloring for $G$ using at least $k$ colors. 
}
\smallskip

First, we design a randomized polynomial time algorithm ${\cal A}_1$ for \pgcstructshort with a success probability at least $(k(d+1))^{-2k^2-k}\cdot 2^{-\ell k}$. We increase the probability of success to a constant by running ${\cal A}_1$ multiple times. Finally, we explain the derandomization of our algorithm.  Then we prove \Cref{thm:pgc} using this algorithm and our structural result (\Cref{thm:deg_reduction}). 
To design the algorithm ${\cal A}_1$, we use the following result of Lokshtanov et al.~\cite{lokshtanov2020covering}.

\begin{proposition}[Lemma 1.1.~\cite{lokshtanov2020covering}] \label{lem:indcover}
There is a linear-time randomized algorithm that, given a $d$-degenerate graph $H$ and an integer $k$, outputs an independent set $Y$ such that for any independent set $X$ in $H$ with $|X|\leq k$, the probability that $X\subseteq Y$ is at least $\left({k(d+1) \choose k} \cdot k(d+1)\right)^{-1}$.
\end{proposition}

\noindent

The algorithm ${\cal A}_1$ has the following steps. 
\begin{enumerate}
\item Color all vertices in $V(G)$ uniformly and independently at random with colors from the set $[k]$. Let the obtained coloring be $\phi : V(G)\rightarrow [k]$, and $Z_i=\phi^{-1}(i)$, for each $i\in [k]$.



\item For each $i\in [\ell]$, let $A_i=(L_i\uplus R_i, E_i)$. For each $j\in [k]$ and $i\in [\ell]$, uniformly at randomly assign $P_{j,i}:=L_i$ or $P_i:=R_i$. 
That is, with probability $\frac{1}{2}$, $P_{j,i}:=L_i$ and with probability $\frac{1}{2}$, $P_{j,i}:=R_i$. Let $D_j=\bigcup_{i\in [\ell]} P_{j,i}$.

\item Now for each $j\in [k]$, we apply the algorithm in ~\Cref{lem:indcover} for $(G[Z_j-D_j], k)$ to obtain an independent set $Y_{j}$. 
\item If $(Y_1,\ldots,Y_k)$ is a $k$-partial Grundy witness of $G$, then 
we output {\bf Yes}, otherwise we output {\bf No}.  

\end{enumerate}

Since the algorithm in ~\Cref{lem:indcover} runs in linear time, 
the algorithm ${\cal A}_1$ can be implemented to run in linear time because $Z_1,\ldots,Z_k$ is a partition of $V(G)$. Clearly, if the algorithm ${\cal A}_1$ outputs {\bf Yes}, then $G$ has a $k$-Partial Grundy witness and hence hence $G$ has a partial Grundy coloring using at least $k$ colors. Next we prove that if $G$ has a $k$-Partial Grundy witness the algorithm ${\cal A}_1$ outputs {\bf Yes} with sufficiently high probability. 

\begin{lemma}
\label{lem:A1proof}
Suppose $G$ has a $k$-Partial Grundy witness. Then the algorithm ${\cal A}_1$ outputs {\bf Yes} with probability at least $(k(d+1))^{-2k^2-k}\cdot 2^{-\ell k}$. 
\end{lemma}

\begin{proof}

There is a partial $k$-Grundy witness $(X_1,\ldots,X_k)$ of $G$ such that $ |X_i| \leq k$, for all $i\in [k]$. This follows from the assumption that $G$ has a $k$-partial Grundy witness and \Cref{obs:k-witness}. 
We say that the $k$-partial Grundy witness $(X_1,X_2,\ldots,X_k)$ is {\em valid} with respect to $\phi$, if $X_i\subseteq Z_i=\phi^{-1}(i)$ for all $i\in [k]$. Since $\sum_{i\in [k]}|X_i|\leq k^2$, the probability that the $k$-partial Grundy witness $(X_1,X_2,\ldots,X_k)$ is valid with respect to $\phi$ is at least ${k^{-k^2}}$;
and we name it as event ${\mathcal E}_c$.  

Now consider Step~2. For each $j\in [k]$ and $i\in [\ell]$, $X_j$ intersects with at most one of $L_i$ or $R_i$, because $A_i$ is a complete bipartite graph. Let ${\mathcal E}_j$ be the event that $X_j$ does not intersects with $D_j$. Thus, $\Pr[{\mathcal E}_j]\geq 2^{\ell}$. Moreover, if the events ${\mathcal E}_c,{\mathcal E}_1,\ldots,{\mathcal E}_k$ happen, then $X_j\subseteq Z_j\setminus D_j$ for all $j\in [k]$. 
The probability that all the events ${\mathcal E}_c,{\mathcal E}_1,\ldots,{\mathcal E}_k$ happen is at least $k^{-k^2}\cdot 2^{-\ell k}$. 

Fix an index $j\in [k]$ and recall that $|X_j| \leq k$. Thus, from~\Cref{lem:indcover}, given $X_j \subseteq Z_j\setminus D_j$, the probability that $X_j \subseteq Y_{j}$ is at least $\left({k(d+1) \choose k} \cdot k(d+1)\right)^{-1}$. Thus the probability that 
$X_j\subseteq Y_j$ for all $j\in [k]$ is at least 

$$\left({k(d+1) \choose k} \cdot k(d+1)\right)^{-k}\cdot k^{-k^2}\cdot 2^{-\ell k}\geq (k(d+1))^{-2k^2-k}2^{-\ell k}.$$
This implies that with probability at least $(k(d+1))^{-2k^2+k}\cdot 2^{-\ell k}$, the algorithm outputs ${\bf Yes}$ in Step~4. This completes the correctness proof. 
%
\end{proof}

Now to increase the probability of success to $2/3$ we run ${\cal A}_1$ $3\cdot (k(d+1))^{2k^2+k}2^{\ell k}$ times and outputs {\bf Yes} if at least one of them outputs {\bf Yes}, and outputs {\bf No}, otherwise. 
That is, we get the following theorem. 

\begin{theorem}
\label{thm:strucalinput}
There is a randomized algorithm for \pgcstructshort running in time $\OO( (k(d+1))^{2k^2+k}\cdot 2^{\ell k} \cdot (m+n))$. In particular, if $(G,k)$ is a 
no-instance then with probability $1$ the algorithm outputs {\bf No}; and  if $(G,k)$ is a yes-instance then with probability $2/3$ the algorithm outputs {\bf Yes}. 
\end{theorem}

Theorems \ref{thm:deg_reduction} and \ref{thm:strucalinput} implies the following theorem. 

\begin{theorem}
\label{thm:genrandomzied}
There is a randomized algorithm for \pgc\ running in time $2^{\OO(k^4)} n^{\OO(1)}$. with the following specification. In particular, if $(G,k)$ is a 
no-instance then with probability $1$ the algorithm outputs {\bf No}; and  if $(G,k)$ is a yes-instance then with probability $2/3$ the algorithm outputs {\bf Yes}. 
\end{theorem}

\begin{proof}[Proof Sketch]
First we run the algorithm mentioned in \Cref{thm:deg_reduction}. If it concludes that $G$ has a partial Grundy coloring with at least $k$ colors, then we output {\bf Yes}. Otherwise, we get  at most $2k^3$ induced bicliques $A_1\cdots,A_{\ell}$
in $G$ such that the following holds. For any $v \in V(G)$, the degree of $v$ in $G-F$ is at most $k^3$, where $F$ is the union of the edges in the above bicliques. That is the degeneracy of $G-F$ is at most $k^3$. Then, we apply \Cref{thm:strucalinput}, and ouput accordingly. 
\end{proof}

\Cref{thm:strucalinput} on $d$-degenerate graphs (where $\ell=0$) implies the following theorem. 

\begin{theorem}
\label{thm:degenerateinput}
There is randomized algorithm for \pgc\ on $d$-degenerate graphs running in time $\OO( (k(d+1))^{2k^2+k}\cdot (m+n))$.  In particular, if $(G,k)$ is a 
no-instance then with probability $1$ the algorithm outputs {\bf No}; and  if $(G,k)$ is a yes-instance then with probability $2/3$ the algorithm outputs {\bf Yes}. 
\end{theorem}




\subsubsection{Derandomization}

Now, we explain how to derandomize \Cref{thm:strucalinput}. Recall the algorithm ${\cal A}_1$ designed earlier in this section. 
The Steps~1-3 are random steps. To derandomize Step~1 we use the concept of ``universal sets''. The Step~2 can be derandomized easily, because the number of potential sets for the choice of $D_j$ is $2^{\ell}$. To derandomize Step~3, we use the notion of independence covering families.

\smallskip
\noindent
\textbf{Universal Sets.} Panolan \cite{fahad2015dynamic} defined an $(n,p,q)$-universal set that is a generalization of an $(n,k)$-universal set given by Naor et al. \cite{naor1995splitters}. 
For $n,p,q \in \mathbb{N}$, an {\em $(n,p, q)$-universal set} is a set of vectors $V \subseteq [q]^n$ such that for any index set $S\subseteq [n]$ of size at most $p$, the set of $p$-dimensional vectors $V|_S = \{\wvec{v}|_S \mid \wvec{v}\in V\}$, contains all the $q^p$ vectors of dimension $p$ with entries from $[q]$.  

\begin{proposition}[\cite{fahad2015dynamic}]\label{prop:universal-set}
Given positive integers $n,p,q$, there is an algorithm
that constructs an $(n, p, q)$-universal set of cardinality $q^p \cdot p^{\OO(\log p)} \cdot\log^2 n$ in time bounded by $q^p \cdot p^{\OO(\log p)}\cdot n\log^2 n$.
\end{proposition}

The next observation immediately follows from the above result.  

\begin{observation}\label{obs:derandom}
Consider numbers $n,p,q\in \mathbb{N}$, where $p,q \leq n$. We can construct a family $\C{Q}$ of functions from $[n]$ to $[q]$ with at most $q^p \cdot p^{\OO(\log p)}\cdot\log^2 n$ functions in time bounded by $q^p \cdot p^{\OO(\log p)}\cdot n\log^2 n$, such that the following holds: for any $X \subseteq [n]$ of size at most $p$ and a function $\chi: X \rightarrow [q]$, there exists a function $f \in \C{Q}$ such that $f|_{X} = \chi$.  
\end{observation} 

\smallskip
\noindent
\textbf{Independence Covering Family.}
To derandomize the algorithm in
\Cref{lem:indcover}, we need the following definition and a known result.

\begin{definition}[$k$-Independence Covering Family~\cite{lokshtanov2020covering}]{\rm
For a graph $G$ and $k\in {\mathbb N}$, a family of independent sets of $G$ is a {\em $k$-{\icf}} for $(G,k)$, denoted by $ \ifam{G}{k}$, if for any independent set $X$ in $G$ of size at most $k$, there exists $Y \in \ifam{G}{k}$ such that $X \subseteq Y$.
}\end{definition}

\begin{proposition}[Lemma 3.2~\cite{lokshtanov2020covering}]
\label{lemma:discl}
There is an algorithm that given a $d$-degenerate graph $G$ on $n$ vertices and $m$ edges, and $k\in {\mathbb N}$, in time bounded by $\OO(\binom{k(d+1)}{k} \cdot 2^{o(k(d+1))} \cdot (n+m)\log n)$, outputs a $k$-{\icf} for $(G,k)$ of size at most $\binom{k(d+1)}{k} \cdot 2^{o(k(d+1))} \cdot \log n$.
\end{proposition}

Now, we are ready to derandomize the algorithm ${\cal A}_1$. To do that, we would derandomize the random steps of the algorithm: Steps 1-3.

Without loss of generality, we can assume that $V(G)=[n]$. For the first step, we can invoke~\Cref{obs:derandom} with $p= k^2$ and $q =k$, and compute a family of at most $k^{k^2} \cdot k^{\OO(\log k)}\cdot\log^2 n$ functions $\C{Q}$ from $[n]$ to $[k]$ in time bounded by $k^{k^2} \cdot k^{\OO(\log k)}\cdot n \log^2 n$ with the following guarantee. For any fixed small $k$-partial Grundy witness $(X_1,X_2, \cdots, X_k)$ in $G$ (if it exists), there is a function $\chi \in \C{Q}$ such that for each $i \in [k]$ and $v \in X_i$, we have $\chi(v) =i$. That is, $(X_1,X_2, \cdots, X_k)$ is valid with respect to $\chi$. 

To derandomize Step~2, let ${\cal D}^{\star}$ be the family of independent sets in $G[A_1\cup \ldots \cup A_{\ell}]$ such that for each $D\in {\cal D}^{\star}$, and $i\in [\ell]$ either $L_i \subseteq D$ or $R_i\subseteq D$, but not both. The number of sets in ${\cal D}^{\star}$ is $2^{\ell}$. Notice that each $D_j$ constructed in Step~2 belongs to ${\cal D}^{\star}$


Now for each $\phi \in \C{Q}$ and 
$D_1,\ldots,D_k\in {\cal D}^{\star}$, we do the following instead of Steps 3 and 4. For each $j\in [k]$, let $Z_j=\phi^{-1}(j)$. Construct a $k$-{\icf}
$\mathcal{F}(G[Z_j\setminus D_j],k)$ that covers all independent sets in $Z_i\setminus D_j$ of size at most $k$, for each $j\in [k]$. Clearly, if $X_j \subseteq Z_j\setminus D_j$, then there is an independent set in $\mathcal{F}(G[Z_j\setminus D_j],k)$ that contains $X_j$, for $j \in [k]$. Thus, if there is a tuple in $\mathcal{F}(G[Z_1\setminus D_1],k) \times \mathcal{F}(G[Z_2\setminus D_2],k) \times \ldots \times \mathcal{F}(G[Z_k\setminus D_k],k)$ which is a $k$-partial Grundy witness, then 
$G$ has a partial Grundy coloring with at least $k$ colors (See Observation~\ref{obs:k-witness-supset}). 
If for each $\phi \in \C{Q}$, and 
$D_1,\ldots,D_k\in {\cal D}^{\star}$
we fail to find a $k$-partial Grundy witness, then we output {\sf No}.

%
%

The correctness of the above procedure follows from~\Cref{obs:derandom},~\Cref{lemma:discl}, and the proof of~\Cref{lem:A1proof}. Let $T_1=k^{k^2}k^{\OO(\log k)}\cdot \log^2 n$ and $T_2=\binom{k(d+1)}{k} \cdot 2^{o(k(d+1))} \cdot \log n$.
The time required to compute the family $\C{Q}$ is bounded by $\C{O}(T_1 \cdot n)$, and the number of functions in $\C{Q}$ is at most $T_1$. 
For a fixed coloring $\phi$, and $D_1,\ldots,D_k\in {\cal D}^{\star}$, the time taken to construct a $k$-\icf~ $\mathcal{F}(G[Z_j\setminus D_j],k)$ is $\OO(T_2 \cdot (n+m))$, for each $j\in [k]$. 
For a fixed coloring $\phi$, and $D_1,\ldots,D_k\in {\cal D}^{\star}$, the running time to check 
if there is a tuple in $\mathcal{F}(G[Z_1\setminus D_1],k) \times \mathcal{F}(G[Z_2\setminus D_2],k) \times \ldots \times \mathcal{F}(G[Z_k\setminus D_k],k)$ which is a $k$-partial Grundy witness, is 
$\left( \Pi_{i=1}^{k} |\mathcal{F}(G[Z_i\setminus D_i],k)|\right)\cdot n^{\OO(1)}=T_2^k\cdot n^{\OO(1)}=2^{\OO(k^2d)} n^{\OO(1)}$, because $(\log n)^k$ is upper bounded by $\OO(2^{k\log k} n)$.\footnote{If $n \leq k^k$,  then $\log n \leq k\log k$ and $(\log n)^k\in \OO(2^{k\log k} n)$.  Otherwise,  $k \in \OO(\log n/\log \log n)$, and $(\log n)^k \in n^{\OO(1)}$ (see, for instance, Ex. 3.18 and its hint from~\cite{cygan2015parameterized}).} Therefore, the total running time is at most $T_1 \cdot 2^ {\ell k}\cdot 2^{\OO(k^2d)} n^{\OO(1)}$, which is upper bounded by  
$ \cdot 2^{\ell k+k^2\log k} \cdot 2^{\OO(k^2d)}\cdot n^{\OO(1)}$. 
Thus, we get the following Theorem.

\begin{theorem}
\label{thm:strucalinputdet}
There is an algorithm  for \pgcstructshort running in time $2^{\OO(k^2d)}\cdot 2^{\ell k+k^2\log k} \cdot n^{\OO(1)}$. 
\end{theorem}

\Cref{thm:pgc} is a direct corollary of Theorems~\ref{thm:strucalinputdet}  and 
\ref{thm:deg_reduction}. For $d$-degenerate graphs, we get the following deterministic algorithm as a corollary of \Cref{thm:strucalinputdet}, where $\ell=0$. 

\begin{theorem}
\label{thm:degeneracydet}
There is an algorithm for \pgc\ that runs in time $2^{\OO(k^2(d+\log k))} n^{\OO(1)}$, where $d$ is the degeneracy of the input graph.
\end{theorem}

\section{FPT Algorithm for Grundy Coloring on $K_{ij}$-free Graphs}
This section aims to prove Theorem~\ref{thm:KijGrundy}. Consider fixed $i,j \in \mathbb{N}\setminus \{0\}$, where $i \geq j$. Recall that a graph is $K_{i,j}$-free if it does not contain a subgraph isomorphic to $K_{i,j}$ as a subgraph. Let $(G,k)$ be an instance of \KijGrundy, where $G$ is a $K_{i,j}$-free graph. We begin by intuitively explaining the flow of our algorithm. 


Consider a Grundy coloring $c: V(G)\rightarrow [k']$ of $G$, where $k' \geq k$, and for each $z \in [k]$, $c^{-1}(z) \neq \emptyset$. Furthermore, for $z \in [k']$, let $C_z = c^{-1}(z)$. 
Let us focus on the first $k$ color classes, and for $z\in [k]$, arbitrarily fix a vertex $v_z \in C_z$. (Note that $v_z$ has a neighbor in $C_{z'}$, for each $z'\in [z-1]$.) We next intuitively describe construction, for each $z\in [k]$, a set $W_z$ initialized to $\{v_z\}$ as follows. Basically, for each $v_z$, add an arbitrarily chosen neighbor of it in color class $C_{z'}$, for every $z'<z$. We do the above process exhaustively; whenever we add a vertex to a set $W_{z}$, we add an arbitrarily chosen neighbor of it from each color class $C_{z'}$ to $W_{z'}$, where $z'<z$. Then, let $W = \cup_{z \in [k]} W_z$; we will call such a set $W$ a {\em $k$-Grundy set} for $G$ and we will show that such a set of size at most $2^{k-1}$ exists (for yes instances). For $z \in [k]$, let $W_{\leq z} = \cup_{z'\in [z]} W_{z'}$ and $W_{>z} = \cup_{z'\in [k]\setminus [z]} W_{z'}$. Note that $c|_W$ is a $k$-Grundy coloring of $G[W]$. Also, we will show that $G$ has a Grundy coloring using at least $k$ colors if and only if some subgraph of $G$ that has a Grundy coloring using {\em exactly} $k$ colors. We remark that the above result and the existence of $W$ are the same as the results of Gy{\'{a}}rf{\'{a}}s et al.~\cite{DBLP:journals/siamdm/GyarfasKL99} and Zaker~\cite{zaker2006results}, although, for the sake of convenience, we state it here slightly differently. If we can identify all the vertices in $W$ (or some other $k$-Grundy set), then we will be done. The first step of our algorithm will be to use the technique of color coding to randomly color the vertices of $G$ using $k$ colors so that, for each $z\in [k]$, $v\in W_z$ is colored $z$; such a coloring will be a {\em nice coloring} and it will be denoted by $\chi$.

The next step of our algorithm is in-spirit, inspired by the design of FPT algorithms based on computations of ``representative sets''~\cite{DBLP:journals/tcs/Marx09,DBLP:journals/jacm/FominLPS16}. To this end, we will interpret $W$ in a ``tree-like'' fashion. With this interpretation, in a bottom-up fashion, for each $z\in [k]$ and $v \in X_z$, we will compute a family $\C{F}'_{z,v}$, so that, if $v \in W_{z}$, then there will be  $W' \in \C{F}'_{z,v}$ so that $W' \cup W_{> z}$ is also a $k$-Grundy set for $G$. We will now formalize the above steps.
\vspace{-3mm}
\subparagraph*{Grundy Tree \& Grundy Witness.}

We begin by reinterpreting/defining a Grundy witness.


\begin{definition}\label{def:grundy-tree-new}{\rm
For each $k \in \mathbb{N}\setminus \{0\}$, we (recursively) define a pair $(T_k, \ell_k: V(T_k) \rightarrow [k])$, called a {\em $k$-Grundy tree}, where $T_k$ is a tree and $\ell_k$ is a {\em labelling} of $V(T_k)$, as follows\longversion{ (see Fig.~\ref{fig:GrundySet-tree} (ii))}:
\begin{enumerate}
\item $T_1 = (\{r_1\}, \emptyset)$ is a tree with exactly one vertex $r_1$ (which is also its root), and $\ell_1(r_1) = 1$.

\item Consider any $k \geq 2$, we (recursively) obtain $T_{k}$ as follows. For each $z \in [k-1]$, let $(T_z,\ell_z)$ be the $z$-Grundy tree with root $r_z$. We assume that for distinct $z,z' \in [k-1]$, $T_z$ and $T_{z'}$ have no vertex in common, which we can ensure by renaming the vertices.\footnote{For the sake of notational simplicity we will not explicitly write the renaming of vertices used to ensure pairwise vertex disjointness of the trees. This convention will be followed throughout this section.} We set $V(T_k) = \big(\cup_{z \in [k-1]}V(T_{z})\big) \cup \{r_k\}$ and $A(T_k) = \big(\cup_{z \in [k-1]}A(T_{z})\big) \cup \{(r_k, r_{z}) \mid z \in [k-1]\}$. We set $\ell_k(r_k) = k$, and for each $z \in [k-1]$ and $t \in V(T_z)$, we set $\ell_k(t) = \ell_z(t)$.
\end{enumerate}
For $v \in V(T_k)$, $\ell_k(v)$ is the {\em label} of $v$ in $(T_k, \ell_k)$, and the elements in $[k]$ are {\em labels} of $T_k$.
}\end{definition}

\shortversion{\begin{observation}[$\spadesuit$]}
\longversion{\begin{observation}}
\label{obs:grundy-pair-size}
For $k \in \mathbb{N}\setminus \{0\}$, for the $k$-Grundy tree $(T_k,\ell_k)$, $|V(T_k)| = 2^{k-1}$.
\end{observation}
\longversion{
\begin{proof}
Note that $|V(T_1)| = 1 = 2^{0}$, and for each integer $k \geq 2$, $|V(T_k)| = 1 + \sum_{z\in [k-1]} |V(T_{z})|$. Thus we can obtain that $|V(T_k)| = 1+ (2^0 + 2^1 + \cdots + 2^{k-2}) = 1+ (2^{k-1} -1)= 2^{k-1}$.
\end{proof}
}

\shortversion{\begin{observation}[$\spadesuit$]}
\longversion{\begin{observation}}
\label{obs:Grundy-pair-label-count}
Consider $k \in \mathbb{N} \setminus \{0\}$ and the $k$-Grundy tree $(T_k,\ell_k)$. We have $|\ell^{-1}(k)| = 1$ and for each $z \in [k-1]$, $|\ell^{-1}_k(z)| = 2^{k-z-1}$.
\end{observation}
\longversion{
\begin{proof}
The claim trivially follows for $z=k$ and $z=k-1$. Suppose for some $z' \in \{2,3,,\cdots, k-1\}$, the claim is true for all $z''\in \{z',z'+1, \cdots, k\}$, and we now prove the statement for $z=z'- 1$. Note that each $t \in V(T_k)$, where $\ell_k(t) > z$, $t$ has exactly one neighbor with label $z$. Moreover, the above describes all the vertices with the label $z$. Thus, we can obtain that $|\ell_k^{-1}(z)| \leq 1+\sum_{z'' \in \{z',z'+1, \cdots, k-1\}}2^{k-z''-1} = 1+ 2^0+2^1+ \cdots +2^{k-z'-1} = 2^{k-z'} = 2^{k-z-1}$.
\end{proof}
}

Next, we define the notion of {\em $k$-Grundy witness} in a graph $G$.

\begin{definition}\label{def:grundy-witness-uncol}{\rm
Consider $k \in \mathbb{N}\setminus \{0\}$ and a graph $G$. A {\em $k$-Grundy witness} for $G$ is a function $\omega: V(T_k) \rightarrow V(G)$, where $(T_k,\ell_k)$ is the $k$-Grundy tree, such that: 1) for each $z \in [k]$, $\{\omega(t) \mid t \in V(T_k) \mbox{ and } \ell_k(t) = z\}$ is an independent set in $G$, 2) for each $t,t' \in V(T_k)$, if $\ell_k(t) \neq \ell_k(t')$ then $\omega(t) \neq \omega(t')$, and 3) for each $(t,t') \in A(T_k)$, we have $\{\omega(t), \omega(t')\} \in E(G)$.
%
%
%
}\end{definition}

\longversion{
\begin{figure}[t]
\centering
\includegraphics[scale=0.8]{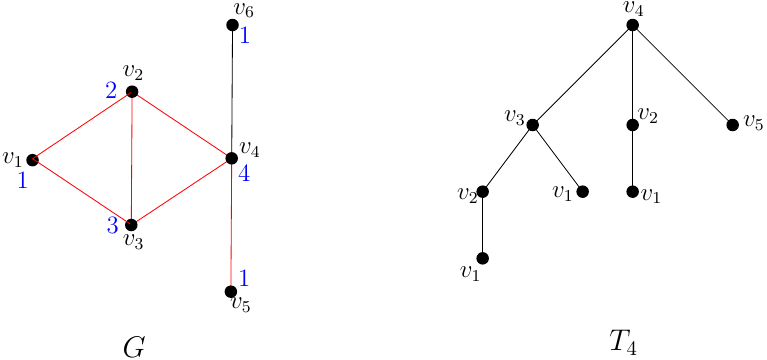}
\caption{An illustration of a graph $G$ that admits a $4$-Grundy coloring (on the left) and a $4$-Grundy-witness $\omega$ (on the right).}
\label{fig:Grundy-set}
\end{figure}
}

Recall that for $k \in \mathbb{N}\setminus \{0\}$, for the $k$-Grundy tree $(T_k, \ell_k)$, $T_k$ is the tree obtained by adding a root vertex $r_k$ attached to the roots of (pairwise vertex disjoint) trees $T_{k-1}, T_{k-2}, \cdots, T_1$, where for each $z \in [k-1]$, $(T_z,\ell_z)$ is the $z$-Grundy tree. We have the following observation.

\shortversion{\begin{observation}[$\spadesuit$]}
\longversion{\begin{observation}}
\label{obs:grundy-witness-closure}
Consider $k \in \mathbb{N}\setminus \{0\}$, a graph $G$ and a $k$-Grundy witness $\omega: V(T_k) \rightarrow V(G)$ for $G$. For each $z \in [k]$, $\omega|_{V(T_{z})}$ is a $z$-Grundy witness for $G$.
\end{observation}
\longversion{
\begin{proof}
Let $\omega: V(T_k) \rightarrow V(G)$ be a $k$-Grundy witness for $G$. The claim trivially follows for $z =k$. Suppose that for some $z' \in \{2,3,\cdots k\}$, the claim is true for all $z\geq z'$, where $z\leq k$. Next, we prove the statement for $z=z'-1$. Note that $\omega'=\omega|_{V(T_{z+1})}$ is a $z$-Grundy witness for $G$. Also, $\omega'|_{V(T_{z})} = \omega|_{V(T_{z})}$.

By Definition \ref{def:grundy-witness-uncol}, for each $\what{z} \in [z+1]$, $\{\omega'(t) \mid t \in V(T_{z+1}) \mbox{ and } \ell_{z+1}(t) = \what{z}\}$ is an independent set in $G$. Thus, for each $\what{z} \in[z]$, $\{\omega(t) \mid t \in V(T_z) \mbox{ and } \ell_z(t) = \what{z}\}$ is an independent set in $G$ (see Definition \ref{def:grundy-tree-new}). Since $T_z$ is a subtree of the tree $T_k$, we can obtain that item $2$, and $3$ of Definition \ref{def:grundy-witness-uncol} holds. Hence, $\omega|_{V(T_{z})}$ is a $z$-Grundy witness for $G$.
\end{proof}
}

The next observation is a partial Grundy counterpart of~\Cref{prop:extension}.


\shortversion{\begin{observation}[$\spadesuit$]}
\longversion{\begin{observation}}
\label{obs:induced-grundy-col}
Consider a graph $G$, any induced subgraph $\what{G}$ of it, and an integer $k \in \mathbb{N}$. If $\what{G}$ has a Grundy coloring that uses exactly $k$ colors, then $G$ has a Grundy coloring that uses at least $k$ colors.
\end{observation}
\longversion{
\begin{proof}
Let $\what{c}: V(\what{G}) \rightarrow [k]$ be a Grundy coloring of $\what{G}$ with exactly $k$ colors. We construct a Grundy coloring $c: V(G) \rightarrow \mathbb{N}$ of $G$ using at least $k$ colors as follows. For each $v \in V(\what{G})$, set $c(v) := \what{c}(v)$. Let $v_1,v_2, \cdots, v_{n'}$ be an arbitrarily fixed order of vertices in $V(G) \setminus V(\what{G})$. Let $G_0 =\what{G}$, and for each $p \in [n']$, $G_p = G[V(\what{G} \cup \{v_1,v_2, \cdots, v_p\})]$. We will iteratively create a Grundy coloring $c_p$ of $G_p$ using at least $k$ colors (in increasing values of $p$) as follows. Note that $c_0=\what{c}$ is already a Grundy coloring of $G_0$ that uses at least $k$ colors. Consider $p \in [n'] \setminus \{0\}$, and assume that we have already computed a Grundy coloring $c_{p-1}: V(G_{p-1}) \rightarrow \mathbb{N}$ of $G_{p-1}$ that uses $k' \geq k$ colors. For each $z \in [k']$, let $V_z = c_{p-1}^{-1}(z)$. For each $v \in V(G_{p-1})$, we set $c_p(v) := c_{p-1}(v)$. If the vertex $v_p$ has a neighbor in each of sets $V_1,V_2, \cdots, V_{k'}$, i.e., if for each $z \in [k']$, $N_G(v_p) \cap V_z \neq \emptyset$, then set $c_p(v_p) := k'+1$. Otherwise, let $z^* \in [k']$ be the smallest number such that $N_G(v_p) \cap V_{z^*} = \emptyset$, and set $c_p(v_p) := z^*$. Notice that by construction, $c_p$ is a Grundy coloring of $G_p$ using at least $k$ colors. From the above discussions, we can obtain that $c_{n'}$ is a Grundy coloring of $G=G_{n'}$ using at least $k$ colors. 
\end{proof}
}

In the following two lemmas, we show that the existence of a $k$-Grundy witness for a graph is equivalent to the graph admitting a Grundy coloring with at least $k$ colors.

\shortversion{\begin{lemma}[$\spadesuit$]}
\longversion{\begin{lemma}}
\label{lem:grundy-eq-witness-one}
For any $k \in \mathbb{N} \setminus \{0\}$ and a graph $G$, if $G$ has a $k$-Grundy witness, then $G$ has a Grundy coloring with at least $k$ colors.
\end{lemma}
\longversion{
\begin{proof}
Consider a graph $G$ and any $k \in \mathbb{N} \setminus \{0\}$. For a $k$-Grundy witness $\omega: V(T_k) \rightarrow V(G)$ of $G$, let $\what{V}_{\omega} = \{\omega(t) \mid t \in V(T_k)\}$, and for each $z \in [k]$, let $\what{V}_{\omega,z} = \{\omega(t) \mid t \in V(T_k) \mbox{ and } \ell_k(t) = z\}$. Note that from item 2 of Definition~\ref{def:grundy-witness-uncol}, $\what{V}_{\omega,1}, \what{V}_{\omega,2}, \cdots, \what{V}_{\omega,k}$ is a partition of $\what{V}_\omega$, where none of the parts are empty. Let $c_\omega: \what{V}_\omega \rightarrow [k]$ be the function such that for each $z \in [k]$ and $v \in \what{V}_{\omega,z}$, we have $c_\omega(v) = q$.

For each $k \in \mathbb{N} \setminus \{0\}$ and a $k$-Grundy witness $\omega: V(T_k) \rightarrow V(G)$ of $G$, we will prove by induction (on $k$) that $c_\omega$ is Grundy coloring of $G[\what{V}_\omega]$ using $k$ colors. The above statement, together with Observation~\ref{obs:induced-grundy-col}, will give us the desired result.

The base case is $k=1$, where $T_1$ has exactly one vertex, $r_1$. For any $1$-Grundy witness, $\omega$ of $G$, note that $c_\omega(\omega(r_1)) = 1$ is a Grundy coloring for $G[\{\omega(r_1)\}]$ using $1$ color. Now for the induction hypothesis suppose that for some $\what{k} \in \mathbb{N} \setminus \{0,1\}$, for each $0< k < \what{k}$, the statement is true. Now we will prove the statement for $k = \what{k}$, and to this end, we consider a $k$-Grundy witness $\omega: V(T_k) \rightarrow [k]$, where $r_k$ is the root of $T_k$. Recall that $T_k$ is the tree obtained by adding a root vertex $r_k$ attached to the roots of (pairwise vertex disjoint) trees $T_{k-1}, T_{k-2}, \cdots, T_1$, where for each $z\in [k-1]$, $(T_z,\ell_z)$ is the $z$-Grundy tree, and $T_z$ is rooted at $r_z$. Let $V' = \what{V}_\omega \setminus \what{V}_{\omega,k}$, and consider a vertex $v \in \what{V}_{\omega,z^*}$, where $z^* \in [k-1]$. We will argue that for each $z'\in [z^*-1]$, $N_G(v)\cap \what{V}_{\omega, z'} \neq \emptyset$.
Note that there must exists $z \in [k-1]$ and $t \in V(T_z)$ such that $\omega(t) = v$ and $\ell_z(v) = z^*$, and we arbitrarily choose one such $z$ and $t$. Let $V_z = \{\omega(t) \mid t \in V(T_z)\}$. From Observation~\ref{obs:grundy-witness-closure}, $\omega_z = \omega|_{V(T_z)}$ is a $z$-Grundy witness for $G$. Thus, from our induction hypothesis, $c_{\omega_z} = c_\omega|_{V_z}$ is a Grundy coloring for $G[V_z]$. From the above we can conclude that for each $q'\in [z^*-1]$, $N_G(v) \cap \what{V}_{\omega,z'} \neq \emptyset$. Now consider the vertex $\omega(r_k) = v^*_k$ and any $z \in [k-1]$. Note that $\ell_{k}(r_z) = z$ and from item 3 of Definition~\ref{def:grundy-witness-uncol}, we can obtain that $\{v^*_k, \omega(r_z)\} \in E(G)$. From the above discussions, we can obtain that $c_\omega$ is Grundy coloring of $G[\what{V}_\omega]$ using at least $z$ colors. This concludes the proof.
%
%
%
%
%
%
\end{proof}
}

\shortversion{\begin{lemma}[$\spadesuit$]}
\longversion{\begin{lemma}}
\label{lem:grundy-eq-witness-two}
For any $k \in \mathbb{N} \setminus \{0\}$ and a graph $G$, if $G$ has a Grundy coloring with at least $k$ colors, then $G$ has a $k$-Grundy witness.
\end{lemma}
\longversion{
\begin{proof}
Consider a Grundy coloring $c: V(G) \rightarrow [k']$ of $G$ with $k' \geq k$ colors, and for each $q \in [k']$, let $C_q = c^{-1}(q)$. We construct a Grundy witness $\omega: V(T_k) \rightarrow V(G)$ by processing labels of $T_k$ starting at $k$ and iteratively proceeding to smaller labels as follows while maintaining the below invariants.

\noindent{\em Pre-condition:} When we begin processing a label $q \in [k-1]$, for each $t \in V(T_k)$ with $\ell_k(t) \geq q$, we have fixed the vertex $\omega(t)$.

\noindent{\em Post-condition:} After processing label $q \in [k]$, we have fixed, for each $t \in V(T_k)$ with $\ell_k(t) \geq q$, and $t' \in N_{T_k}[t]$, the vertex $\omega(t')$; and these are the only vertices in $T_k$ for which the vertex in $G$ assigned by $\omega$ is determined.

Note that the pre-condition is vacuously satisfied for $q=k$. Recall that $T_k$ is the tree obtained by adding a root vertex $r_k$ attached to the roots $r_{k-1},r_{k-2}, \cdots, r_1$ of (pairwise vertex disjoint) trees $T_{k-1}, T_{k-2}, \cdots, T_1$, respectively, where for each $q\in [k-1]$, $(T_q,\ell_q)$ is a $q$-Grundy tree. Pick any vertex $v_k \in C_k$, and set $\omega(r_k) := v_k$ and for each $q \in [k-1]$, set $\omega(r_q) := w^k_q$, where $w^k_q$ is an arbitrarily chosen neighbor of $v_k$ from $C_q$ (which exists as $c$ is a Grundy coloring). After the above step, the post-condition is satisfied for $q=k$.


Now we (iteratively, in decreasing order) consider $q \in [k-1]\setminus \{1\}$. From the pre-condition for $q$, we have fixed, for each $t \in V(T_k)$ with $\ell_k(t) \geq q$, the vertex $\omega(t)$. Consider $t \in V(T_k)$ with $\ell_k(t) = q$ and let $v_t = \omega(t)$. Let $\what{T}_q$ be the subtree of $T_k$ rooted at $t$, and let $\what{\ell}_q = \ell_k|_{V(\what{T}_q)}$. Notice that $(\what{T}_q, \what{\ell}_q)$ is a $q$-Grundy tree, where $\what{T}_q$ is the tree obtained from by adding edge between $t$ and the roots $\what{r}_{q-1},\what{r}_{q-2}, \cdots, \what{r}_1$ of $\what{T}_{q-1},\what{T}_{q-2}, \cdots, \what{T}_1$, respectively, where $(\what{T}_{q'}, \ell_k|_{V(\what{T}_{q'})})$ is a $q'$-Grundy tree, for each $q' \in [q-1]$. For each $q' \in [q]$, let $\what{w}^q_{q'}$ be an arbitrarily chosen vertex from $N_G(v) \cap C_{q'}$, and we set $\omega(\what{r}_{q'}) = \what{w}^q_{q'}$. Notice that after the above step, the post-condition is satisfied for $q$, and the pre-condition is satisfied for $q-1$.

After we are done processing each $q \in [k] \setminus \{1\}$, the post-condition for $q=2$ (and the pre-condition of $q=1$ implies that for each $t \in V(T_k)$, we have determined the vertex $\omega(t)$). Moreover, the construction of $\omega$ implies that all the three conditions in Definition~\ref{def:grundy-witness-uncol} are satisfied. This concludes the proof.
\end{proof}
}

We next summarize the result we obtain from the above two lemmas.

\begin{corollary}\label{cor:grundy-witness-eq-coloring}
Consider any $k \in \mathbb{N} \setminus \{0\}$ and a graph $G$. The graph $G$ has a $k$-Grundy witness if and only if $G$ has a Grundy coloring with at least $k$ colors.
\end{corollary}

\noindent{\bf Color Coding of $G$.} We will next use Observation~\ref{obs:induced-grundy-col} and the above corollary to simplify our job in the following sense. Let $\omega: V(T_k) \rightarrow V(G)$ be a (fixed) $k$-Grundy witness of $G$ (if it exists), where $(T_k,k)$ is a $k$-Grundy tree. Let $\what{V}_{\omega} = \{\omega(t) \mid t \in V(T_k)\}$, and for each $q \in [k]$, let $\what{V}_{\omega,q} = \{\omega(t) \mid t \in V(T_k) \mbox{ and } \ell_k(t) = q\}$. Roughly speaking, our new objective will be to find the vertices in $\what{V}_\omega$ and say that $G[\what{V}_\omega]$ admits a Grundy coloring with at least $k$ colors, and then from Observation~\ref{obs:induced-grundy-col}, conclude that $G$ admits a Grundy coloring with at least $k$ colors. We will use the technique of color coding introduced by Alon et al.~\cite{DBLP:journals/jacm/AlonYZ95}, to color the vertices in $\what{V}_{\omega}$ ``nicely'' as follows. Color each vertex in $G$ uniformly at random using a color from $[k]$, and let $\chi: V(G)\rightarrow [k]$ be this coloring. A {\em nice} coloring is the one where, for each $q \in [k]$, the coloring assigns the color $q$ to all the vertices in $\what{V}_{\omega,q}$. 



We will work with the assumption that $\chi$ is a nice coloring of $G$, and for each $q \in [k]$, let $X_q = \chi^{-1}(q)$. Note that we can eliminate the randomness using~\Cref{obs:derandom}. Our objective will be to look for a $k$-Grundy witness $\what{\omega} : V(T_k) \rightarrow V(G)$, where $(T_k,k)$ is a $k$-Grundy tree, such that for each $q \in [k]$ and $t \in V(T_k)$ with $\ell_k(t) = q$, we have $\what\omega(t) \in X_q$. To this end, we will store a ``Grundy representative family'' for each vertex in a bottom-up fashion, starting from $q=1$. The definition of such a representative is inspired by the $q$-representative families~\cite{DBLP:journals/tcs/Marx09,DBLP:journals/jacm/FominLPS16}, although here we need a ``vectorial'' form of representation. To this end, we introduce the following notations and definitions.

\subparagraph{Grundy Representative Sets.} Recall we have the coloring $\chi$ of $G$ with color classes $X_z = \chi^{-1}(z)$, for $z \in [k]$. A vertex subset $A \subseteq V(G)$ is {\em$\chi$-independent} if for each $z \in [k]$, $A \cap X_z$ is an independent set in $G$. For $p \in \mathbb{N}$, a family of vertex subsets $\C{F}$ is a {\em $p$-family} if each set in $\C{F}$ has size at most $p$ and each $A \in \C{F}$ if $\chi$-independent. We will only be working with vectors whose all entries are from $\mathbb{N}$ without explicitly stating it. For a vector $\wvec{q} = (q_1,q_2,\cdots, q_k)$, $\sm{\wvec{q}}$ denotes the number $\sum_{z\in [q]}q_z$. For a vector $\wvec{q} = (q_1,q_2,\cdots, q_k)$ and $B \subseteq V(G)$, we say that the {\em size} of $B$ is $\wvec{q}$, written as $|B| = \wvec{q}$, if for each $z \in [k]$, $|B \cap X_z| = q_z$. For vertex subsets $A$ and $B$, $A$ {\em fits} $B$ if $A \cup B$ is $\chi$-independent. For two vectors $\wvec{q_1} =(q^1_1,q^1_2,\cdots, q^1_k)$ and $\wvec{q_2} =(q^2_1,q^2_2,\cdots, q^2_k)$, and $\diamond \in \{\leq, \geq, >, <,=\}$, we write $\wvec{q_1} \diamond \wvec{q_2}$ if for each $z\in [k]$, we have $q^1_z \diamond q^2_z$. We next define the notion of {$\wvec{q}$-Grundy representation}.


\begin{definition}\label{def:vec-rep}{\rm
Consider $p \in \mathbb{N}$, a vector $\wvec{q} = (q_1,q_2,\cdots, q_k)$, and a $p$-family $\C{F}$ of vertex subsets of $G$. For a sub-family $\C{F}' \subseteq \C{F}$, we say that $\C{F}'$ {\em $\wvec{q}$-Grundy represents} $\C{F}$, written as $\C{F}' \subseteq^{\wvec{q}}_{grep} \C{F}$, if the following hold. For any set $B$ of size $\wvec{q}$, if there is $A \in \C{F}$ that fits $B$, then there is $A' \in \C{F}'$ that fits $B$. In the above, $\C{F}'$ is a {\em $\wvec{q}$-Grundy representative} for $\C{F}$.
}\end{definition}

Next, we obtain some properties regarding $\wvec{q}$-Grundy representatives.

\shortversion{\begin{observation}[$\spadesuit$]}
\longversion{\begin{observation}}
\label{obs:transitive-vec-rep}
Consider $p \in \mathbb{N}$, a vector $\wvec{q} = (q_1,q_2,\cdots, q_k)$, and any two $p$-families $\C{F}_1$ and $\C{F}_2$. If $\C{F}'_1 \subseteq^{\wvec{q}}_{grep} \C{F}_1$ and $\C{F}'_2 \subseteq^{\wvec{q}}_{grep} \C{F}_2$, then $\C{F}'_1 \cup \C{F}'_2 \subseteq^{\wvec{q}}_{grep} \C{F}_1 \cup \C{F}_2$.
\end{observation}

\longversion{
\begin{proof}
Consider a set $B$ of size $\wvec{q}$ such that there is some set $A\in \C{F}_1\cup \C{F}_2$ that fits $B$. If $A\in \C{F}_1$, then there exists some $A'\in \C{F}'_1$ that fits $B$ as $\C{F}'_1 \subseteq^{\wvec{q}}_{grep} \C{F}_1$. If $A\in \C{F}_2$, then there exists some $A'\in \C{F}'_2$ that fits $B$ as $\C{F}'_2 \subseteq^{\wvec{q}}_{grep} \C{F}_2$. Thus, there exists $A'\in \C{F}'_1\cup\C{F}'_2$ that fits $B$. Hence, $\C{F}'_1 \cup \C{F}'_2 \subseteq^{\wvec{q}}_{grep} \C{F}_1 \cup \C{F}_2$.
\end{proof}
}

Consider $p \in \mathbb{N}$ and $v \in V(G)$. For a family $\C{F}$ over $V(G)$, $\C{F} + v$ denotes the family $\{A\cup \{v\} \mid A \in \C{F} \mbox{ and $A\cup \{v\}$ is $\chi$-independent}\}$. Similarly, $\C{F} - v$ denotes the family $\{A\setminus \{v\} \mid A \in \C{F}\}$. A $p$-family $\C{F}$ is a {\em $(p,v)$-family} if for each $A\in \C{F}$, we have $v \in A$.

\shortversion{\begin{observation}[$\spadesuit$]}
\longversion{\begin{observation}}
\label{obs:transitive-vec-rep-one-vertex}
Consider $p \in \mathbb{N}$, a vector $\wvec{q} = (q_1,q_2,\cdots, q_k)$, a vertex $v \in V(G)$ and a $(p,v)$-family $\C{F}$. Let $\wvec{h}$ be the vector obtained from $\wvec{q}$ by increasing its $\chi(v)$th coordinate by $1$. If $\C{F}' \qrep{\wvec{h}} \C{F}-v$ and $\C{F}'' \qrep{\wvec{q}} \C{F}-v$, then $(\C{F}' + v) \cup (\C{F}'' + v) \qrep{\wvec{q}} \C{F}$.
\end{observation}
\longversion{
\begin{proof}
Consider a set $B$ of size $\wvec{q}$ for which some $A\in \C{F}$ fits $B$. Since $\C{F}$ is a $(p,v)$-family, $v\in A$ and $A\setminus \{v\}\in \C{F}-v$. If $v \notin B$, then $B\cup \{v\}$ is of size $\wvec{h}$, in which case there must exist $A'\in \C{F}'$ that fits $B\cup \{v\}$ as $\C{F}'\qrep{\wvec{h}} \C{F}-v$. Otherwise, $v \in B$, and thus $B\cup \{v\} = B$. For the above case, there is $A'\in \C{F}''$ that fits $B$ as $\C{F}'' \qrep{\wvec{q}} \C{F}-v$. In any of the above two cases, $A'\cup \{v\} \cup B$ is $\chi$-independent, and $A' \cup \{v\} \in (\C{F}' + v) \cup (\C{F}'' + v)$. This concludes the proof.
\end{proof}
}



For a $p_1$-family $\C{F}_1$ and a $p_2$-family $\C{F}_2$, we define a $(p_1+p_2)$-family, $\C{F}_1 \star \C{F}_2 = \{A_1 \cup A_2 \mid A_1 \in \C{F}_1, A_2 \in \C{F}_2, \mbox{ and $A_1 \cup A_2$ is $\chi$-independent}\}$. The following lemma will be helpful in obtaining a $\wvec{q}$-representative for $\C{F}_1 \star \C{F}_2$.



\shortversion{\begin{lemma}[$\spadesuit$]}
\longversion{\begin{lemma}}
\label{lem:rep-compose}
Consider a $p_1$-family $\C{F}_1$, a $p_2$-family $\C{F}_2$, and a vector $\wvec{q} =(q_1,q_2,\cdots, q_k)$, where $\sm{\wvec{q}} + p_1 +p_2 \leq 2^{k-1}$. Let $\C{F}'_1 \subseteq \C{F}_1$ be a $p_1$-family such that for every vector $\wvec{q_1} \geq \wvec{q}$ with $\sm{\wvec{q_1}} \leq \sm{\wvec{q}} +p_2$, $\C{F}'_1 \qrep{\wvec{q_1}} \C{F}_1$. Similarly, consider a $p_2$-family $\C{F}'_2 \subseteq \C{F}_2$ such that for every vector $\wvec{q_2} \geq \wvec{q}$ with $\sm{\wvec{q_2}} \leq \sm{\wvec{q}} +p_1$, $\C{F}'_2 \qrep{\wvec{q_2}} \C{F}_2$. Then, $\C{F}'_1 \star \C{F}'_2 \qrep{\wvec{q}} \C{F}_1 \star \C{F}_2$.
%
\end{lemma}
\longversion{
\begin{proof}
Consider any $B \subseteq V(G)$ of size $\wvec{q}$ for which there is $A \in \C{F}_1 \star \C{F}_2$, such that $A$ fits $B$. As $A \in \C{F}_1 \star \C{F}_2$, there must exist sets $A_1\in \C{F}_1$, $A_2 \in \C{F}_2$, such that $A_1\cup A_2 = A$.


Let $\wvec{\delta_1}= (\delta^1_z = |(A_2 \cap X_z) \setminus B|)_{z \in [k]}$. Note that $|B\cup A_2| = \wvec{q} +\wvec{\delta_1}$, $A_1$ fits $B\cup A_2$ and $\sm{\wvec{q}} +\sm{\wvec{\delta_1}}\leq\sm{\wvec{q}} + p_2$. 
By the premise of the lemma, there exist $A'_1 \in \C{F}'_1$ such that $A'_1$ fits $B \cup A_2$, as $\C{F}'_1 \qrep{\wvec{q}+\wvec{\delta_1}} \C{F}_1$. The above implies that $A_2$ fits $B \cup A'_1$, where $A'_1\in \C{F}'_1$. Let $\wvec{\delta_2} = (\delta^2_z = |(A'_1 \cap X_z) \setminus B|)_{z \in [k]}$, and note that $|B\cup A'_1| = \wvec{q} +\wvec{\delta_2}$, $A_2$ fits $B\cup A'_1$ and $\sm{\wvec{q}} +\sm{\wvec{\delta_2}}\leq\sm{\wvec{q}} + p_1$. Again, as $\C{F}'_2 \qrep{\wvec{q}+\wvec{\delta_2}} \C{F}_2$, there exists $A'_2 \in \C{F}'_2$ such that $A'_2$ fits $B\cup A'_1$. The above discussions imply that, $A'_1 \in \C{F}'_1$, $A'_2 \in \C{F}'_2$, and thus $A'_1\cup A'_2 \in \C{F}'_1\star \C{F}'_2$, where $A'_1\cup A'_2$ fits $B$. This concludes the proof.
\end{proof}
}

Recall that $G$ is a $K_{i,j}$-free graph, where $i \geq j$. Consider any computable function $f(k)$. Let $\eta_{f(k)} := i \cdot f(k) \cdot k$; where we skip the subscript $f(k)$ when the context is clear. Also, for $p\in \mathbb{N}$, let $\alpha_p := \alphaval$; again we skip the subscript $p$, when the context is clear. 
We next state the main lemma, which lies at the crux of our algorithm.

\shortversion{\begin{lemma}[$\spadesuit$]}
\longversion{\begin{lemma}}
\label{lem:grundy-rep}
Consider any computable function $f: \mathbb{N}\rightarrow \mathbb{N}\setminus \{0\}$. There is an algorithm that takes as input $k \in \mathbb{N}\setminus \{0\}$, $p \in \mathbb{N}$, a vector $\wvec{q} = (q_1,q_2,\cdots, q_k)$, and a $p$-family $\C{F}$ of vertex subsets of a $K_{i,j}$-free graph $G$ on $n$ vertices with a coloring $\chi: V(G) \rightarrow [k]$, where $p + \sm{\wvec{q}} \leq f(k)$. In time bounded by $\C{O}(\alpha^{2p + \sm{\wvec{q}}} \cdot p \cdot |\C{F}|)$ we can find $\C{F}' \subseteq \C{F}$ with at most $\alpha^{2p + \sm{\wvec{q}}}$ sets such that $\C{F}' \subseteq^{\wvec{q}}_{\sf grep} \C{F}$.
\end{lemma}

\longversion{We prove the above lemma in~\Cref{sec:Grundy-rep}.} In the remainder of this section, we prove Theorem~\ref{thm:KijGrundy}, assuming the correctness of~\Cref{lem:grundy-rep}.

\subparagraph{\bf Some Useful Notations.} For $z\in \mathbb{N} \setminus \{0\}$ and $z' \in [z]$, let $\gamma_{z,z'}$ be the number of vertices with label $z'$ in the $z$-Grundy tree $(T_z,\ell_z)$, i.e., $\gamma_{z,z'} = |\ell_{z}^{-1}(z')|$ (see~\Cref{obs:Grundy-pair-label-count}).

Let $\wvec{q}^* = \wvec{\gamma_k} := (\gamma_{k,1}, \gamma_{k,2}, \cdots, \gamma_{k,k})$. We will define a vector $\wvec{q_z}^* = (q^*_{z,1}, q^*_{z,2}, \cdots, q^*_{z,k})$, for every $z \in [k]$. Intuitively speaking, the $z'$th entry of $\wvec{q_z}^*$ will denote the number of vertices with label $z'$ appearing in $T_k$ after removing exactly one subtree rooted at a vertex with label $z$. Formally, for each $z' \in \{z+1,z+2, \cdots, k\}$, we have $q^*_{z,z'} = \gamma_{k,z'}$, and for each $z' \in [z]$, $q^*_{z,z'} = \gamma_{k,z'} - \gamma_{z,z'}$. For $z \in [k]$, we let $\wvec{0_z}$ be the vector of dimension $k$ where the $z$th entry is $1$, and all the other entries are $0$.

For a tree $\what{T}$ rooted at $r$ and $t\in V(\what{T})$, we let $\what{T}^t$ be the subtree of $\what{T}$ rooted at $t$, i.e., $V(\what{T}^t) = \{t'\in V(\what{T}) \mid t' = t, \mbox{ or $t'$ is a descendant of $t$ in } \what{T}\}$ and $\what{T}^t = \what{T}[V(\what{T}^t)]$.

For a set $W \subseteq V(G)$, we say that $W$ is a {\em $k$-Grundy set} if there is a $k$-Grundy witness $\omega: V(T_k) \rightarrow W$ for $G$. Moreover, $W$ is {\em minimal} if no proper subset $W' \subset W$ is a $k$-Grundy set for $G$. For a $k$-Grundy set $W$ and a $k$-Grundy witness $\omega: V(T_k) \rightarrow W$ for $G$, for $t \in V(T_k)$, we let $W_{{\sf sub},t} = \{\omega(t') \mid t' \in V(T^t_k)\}$ and $W_{{\sf exc},t}=\{\omega(t') \mid t' \in V(T_k) \setminus V(T^t_k)\}$.

Recall that we have a graph $G$ and a coloring $\chi: V(G) \rightarrow [k]$, where for $z\in [k]$, we have $X_z = \chi^{-1}(z)$. For $z \in [k]$ and $v \in V(G)$, we define $\C{F}_{z,v} := \{W \subseteq \cup_{z' \in [z]} X_{z'} \mid v \in W, W \mbox{ is $\chi$-independent and } z \leq |W| \leq 2^{z-1}\}$.



%


\subparagraph*{Description of the Algorithm.} The objective of our algorithm will be to compute, for each $z \in [k]$ and $v \in X_z$, a family $\C{F}'_{z,v} \subseteq \C{F}_{z,v}$; starting from $z =1$ (and then iteratively, for other values of $z$ in increasing order), satisfying the following constraints:
\begin{description}
\item [Size Constraint.] $|\C{F}'_{z,v}| \leq \alpha^{2^{k}+1}$.
\item [Correctness Constraint.] For any $z\in [k]$ and $v \in X_z$, the following holds:
\begin{enumerate}
\item Each $A \in \C{F}'_{z,v}$ is a $z$-Grundy set in $G$.
\item Consider any minimal $k$-Grundy set $W$, such that $v \in W$ (if it exists). Furthermore, let $\omega: V({T_k}) \rightarrow W$ be a $k$-Grundy witness for $G$. For any $t\in V(T_k)$ with $\omega(t) =v$, where $\wvec{q_t} = |W_{{\sf exc},t}|$, there is $W' \in \C{F}'_{z,v} \subseteq \C{F}_{z,v}$ such that $W_{{\sf exc},t} \cup W'$ is a $k$-Grundy set in $G$, i.e., $\C{F}'_{z,v} \qrep{\wvec{q_t}} \C{F}_{z,v}$.
\end{enumerate}
\end{description}

\noindent\textbf{Base Case:} We are in our base case when $z=1$; note that $[2^{0}] = \{1\}$. For each $v \in X_1$, set $\C{F}'_{1,v} := \C{F}_{1,v} =\{\{v\}\}$. Note that $\C{F}'_{1,v}$ satisfies both the size and the correctness constraints.\\

\noindent\textbf{Recursive Formula.} Consider $z \in [k]\setminus \{1\}$ and $v \in X_z$. 
We suppose that for each $z' \in [z-1]$ and $v' \in X_{z'}$, we have computed $\C{F}'_{z',v'}$ that satisfies both the size and the correctness constraints.


For each $z' \in [z-1]$, we create a family $\C{F}_{z,v,z'}$, initialized to $\emptyset$ as follows. For each $u \in X_{z'} \cap N_G(v)$ and $W \in \C{F}'_{z',u}$, if $W \cup \{v\}$ is $\chi$-independent and $|W \cup \{v\}| \leq 2^{z-1}$, then add $W \cup \{v\}$ to $\C{F}_{z,v,z'}$. Note that $|\C{F}_{z,v,z'}| \leq n \cdot \alpha^{2^{k}+1}_{p}$, where $p = 2^{z'-1}$. Using Lemma~\ref{lem:grundy-rep}, for each vector $\wvec{q} \leq \wvec{q_z}^*$, we compute $\C{F}'_{z,v,z', \wvec{q}} \qrep{\wvec{q}} \C{F}_{z,v,z'}$, where $|\C{F}'_{z,v,z',\wvec{q}}| \leq \alpha^{2^{k}}$, and set $\C{F}'_{z,v,z'} = \cup_{\wvec{q} \leq \wvec{q_z}} \C{F}'_{z,v,z', \wvec{q}}$. Note that $\C{F}'_{z,v,z'} \leq 2^{(k-1)k} \cdot \alpha^{2^{k}} \leq \alpha^{2^{k}+ 1}$ and we can compute it in time bounded by $\C{O}(\alpha^{2^{k}+ 1} \cdot 2^{k-1} \cdot |\C{F}_{z,v,z'}|)$.

Next we will iteratively ``combine and reduce'' the families $\C{F}'_{z,v,z'}$, for $z' \in [z]$, to obtain a family $\what{\C{F}}_{z,v} \subseteq \C{F}_{z,v}$ as follows. We set $\what{\C{F}}_{z,v,1} := \C{F}'_{z,v,1}$. Iteratively, (in increasing order), for each $z' \in [z-1] \setminus \{1\}$, we do the following:
\begin{enumerate}
\item Set $\wtilde{\C{F}}_{z,v,z'} := \what{\C{F}}_{z,v,z'-1} \star \C{F}'_{z,v,z'}$.

\item Compute $\what{\C{F}}_{z,v,z',\wvec{q}} \qrep{\wvec{q}}\wtilde{\C{F}}_{z,v,z'}$, for each $ \wvec{q} \leq \wvec{q_{z'}} = \wvec{\gamma_k} -\big( \sum_{\what{z} \in [z']} (\wvec{\gamma_k} - \wvec{q_{\what{z}}}^*) \big)-\wvec{0_z}$, and set $\what{\C{F}}_{z,v,z'} = \cup_{\wvec{q} \leq \wvec{q_{z'}}} \what{\C{F}}_{z,v,z',\wvec{q}}$. Note that $|\what{\C{F}}_{z,v,z'}| \leq \alpha^{2^{k}+1}$ and it can be computed in time bounded by $\C{O}(\alpha^{2^{k}+ 1} \cdot 2^{k-1} \cdot |\wtilde{\C{F}}_{z,v,z'}|)$.
\end{enumerate}

We add each $A \in \what{\C{F}}_{z,v,z-1}$ to $\C{F}'_{z,v}$, which is a $z$-Grundy set in $G$ (note that since the size of each set is bounded by $2^{k-1}$, we can easily do it in the allowed amount of time). In the following lemma, we show that $\C{F}'_{z,v}$ satisfies the correctness constraints.


\shortversion{\begin{lemma}[$\spadesuit$]}
\longversion{\begin{lemma}}
\label{lem:main-corect-rep}
$\C{F}'_{z,v}$ satisfies the correctness constraint.
\end{lemma}
\longversion{
\begin{proof}
Consider any $v \in X_z$ and a minimal $k$-Grundy set $W$, such that $v \in W$ (if it exists) and let $\omega: V({T_k}) \rightarrow W$ be a $k$-Grundy witness for $G$. Next consider any $t\in V(T_k)$ with $\omega(t) =v$. We will argue that, there is $W' \in \C{F}'_{z,v}$ such that $W_{{\sf exc},t} \cup W'$ is a $k$-Grundy set in $G$. For each $z' \in [z-1]$, let $t_{z'}$ be the child of $t$ in $T_k$ with $\ell_k(t_{z'}) = z'$ and $v_{z'} = \omega(t_{z'})$. Note that for each $z' \in [z-1]$, $v_{z'} \in X_{z'}$. For each $z' \in [z-1]$, let $\what{A}_{z'} = \{\omega(t') \mid t' \in V(T^{t_{z'}}_k)\}$, and note that $|\what{A}_{z'}| \leq 2^{z'-1}$. Now we iteratively take the union of the above sets as follows. For each $z' \in [z-1]$, let $A_{z'} =\cup_{\what{z} \in [z']} \what{A}_{\what{z}}$. Now for each $z' \in [z-1]$, we construct a subset, $B_{z'}$ of $W$ that contains $\omega(t')$, for each $t' \in V(T_k)$ that does not belong to the subtrees rooted at any of the vertices $t_1,t_2, \cdots, t_{z'}$. Formally, for $z' \in [z-1]$, let $B_{z'} = \{\omega(t') \mid t' \in V(T_k) \setminus \big( \bigcup_{z'' \in [z']} V(T^{t_{z''}}_k) \big)\}$. Furthermore, let $\wvec{s_{z'}}$ be the size of $B_{z'}$. Notice that for each $z' \in [z-1]$, all of the following holds:
\begin{enumerate}
\item $\wvec{s_{z'}}\leq \wvec{q_{z'}}$,
\item $|A_{z'}| \leq \sum_{\what{z} \in [z']} 2^{\what{z}-1}$,
\item $|\what{A}_{z'}| \leq 2^{z'-1}$ and $\what{A}_{z'} \in \C{F}_{z',v_{z'}}$,
\item $A_{z'} \cup B_{z'} = W$, and thus, $A_{z'}$ fits $B_{z'}$.
\end{enumerate}

We will now iteratively define sets $A'_1,A'_2,\cdots, A'_{z-1}$ and
functions $\omega_1,\omega_2\cdots, \omega_{z-1}$, and we will ensure that, for each $z' \in [z-1]$, we have: i) $\omega_{z'}: V(T_k) \rightarrow A'_{z'} \cup B_{z'}$ is a $k$-Grundy witness for $G$, ii) for each $z' \in [z-1]$ and $t' \in V(T_k) \setminus \big( \bigcup_{z'' \in [z']} V(T^{t_{z''}}_k) \big)$, we have $\omega_{z'}(t') = \omega(t')$, iii) $A'_{z'} \in \what{\C{F}}_{z,v,z'}$, and iv) for each $z''\in [z']$, there is a minimal $z''$-Grundy set $A'_{z',z''} \subseteq A'_{z'}$, where the unique vertex in $A'_{z',z''} \cap X_{z''}$ is a neighbor of $v$.

Recall that $z \geq 2$ and $\what{\C{F}}_{z,v,1} = \C{F}'_{z,v,1} \subseteq \C{F}_{z,v,1}$. Also, we have $\what{A}_1 = A_1 = \{u'\}$, for some $u' \in N_G(v) \cap X_1$, and $A_1 \cup B_1 = W$ is a $k$-Grundy set. Thus, there must exist $A'_1 \in \C{F}'_{z,v,1}$ such that $A'_1 \cup B_1$ is $\chi$-independent. Moreover by the construction of $\C{F}'_{z,v,1}$, $A'_1 = \{u\}$, for some $u \in N_G(v) \cap X_1$. Let $\omega_{1}: V(T_k) \rightarrow A'_{1} \cup B_{1}$ be the function such that for each $t' \in V(T_k) \setminus \{t_1\}$, we have $\omega_1(t') = \omega(t')$ and $\omega_1(t_1) = u$. As $A'_1 \cup B_1$ is $\chi$-independent and $\{u,v\}\in E(G)$, we can obtain that $\omega_1$ is a $k$-Grundy witness for $G$. Note that if $z=2$, then by the above arguments, we have constructed the desired sets and functions, which is just the set $A'_1$ and the function $\omega_1$.

We now consider the case when $z'\geq 2$. Also, we assume that for some $\what{z} \in [z-2]$, for each $z' \in [\what{z}]$, we have constructed $A'_{z'}$ and $\omega_{z'}$ satisfying the desired condition. Now we prove the statement for $z' = \what{z}+1$. Note that $A'_{z'-1} \cup B_{z'-1}$ is a $k$-Grundy set and $\omega_{z'-1}: V(T_k) \rightarrow A'_{z'-1} \cup B_{z'-1}$ is a $k$-Grundy witness for $G$, where $A'_{z'-1} \in \what{\C{F}}_{z,v,z'}$. Note that $\what{A}_{z'} \subseteq B_{z'-1}$. Let $B'_{z'} = \{\omega_{z'-1}(t') \mid t' \in V(T_k) \setminus V(T_k^{t_{z'}})\}$. 
Note that $|B'_{z'}| \leq \wvec{q_{z'}}$ and $\what{A}_{z'}$ fits $B'_{z'}$, and recall that $\what{A}_{z'} \in \C{F}_{z',v_{z'}}$. As $\C{F}'_{z',v_{z'}} \qrep{\wvec{q}} \C{F}_{z',v_{z'}}$, for every $\wvec{q} \leq \wvec{q_{z'}}$, there must exists $\wtilde{A}_{z'} \in \C{F}'_{z',v_{z'}}$, such that $\wtilde{A}_{z'}$ fits $B'_{z'}$. By the construction of $\C{F}'_{z',v_{z'}}$, we have $v_{z'} \in \wtilde{A}_{z'}$ and $\wtilde{A}_{z'} \cup B'_{z'}$ is $\chi$-independent, and also $v \in B'_{z'}$. From the above discussions we can conclude that $\wtilde{A}_{z'} \cup \{v\} \in \C{F}_{z,v,z'}$. Moreover, as $A'_{z'-1} \in \what{\C{F}}_{z,v,z'-1}$, $A'_{z'-1} \subseteq B'_{z'}$ and $\wtilde{A}_{z'} \cup B'_{z'}$ is $\chi$-independent, we can obtain that $A'_{z'-1} \cup \wtilde{A}_{z'} \cup \{v\} \in \what{\C{F}}_{z,v,z'-1} \star \C{F}'_{z,v,z'} =\wtilde{\C{F}}_{z,v,z'}$. As $\what{\C{F}}_{z,v,z'} \qrep{\wvec{q}} \wtilde{\C{F}}_{z,v,z'}$, for every $\wvec{q} \leq \wvec{q_{z'}}^*$ and $|B_{z'}| \leq \wvec{q_{z'}}^*$, there must exists $A'_{z'} \in \what{\C{F}}_{z,v,z'}$ such that $A'_{z'} \cup B_{z'}$ is $\chi$-independent.

As $A'_{z'} \in \what{\C{F}}_{z,v,z'}$, there must exist $\what{C} \in \what{\C{F}}_{z,v,z'-1}$ and $C' \cup \{v\} \in \C{F}'_{z,v,z'}$, such that $A'_{z'} = \what{C} \cup C' \cup \{v\}$. By the correctness for $z'-1$, $C'$ contains for each $z''\in [z'-1]$, a minimal $z''$-Grundy set $A'_{z'-1,z''} \subseteq A'_{z'-1}$, where the unique vertex in $A'_{z'-1,z''} \cap X_{z''}$ is a neighbor of $v$. Also by the construction of $\C{F}'_{z,v,z'}$, $C'$ contains a minimal $z'$-Grundy set $C''$ in $G$, where the unique vertex in $C'' \cap X_{z'}$ is a neighbor of $v$. From the above discussions, we can conclude that $A'_{z'} \cup B_{z'}$ is a $k$-Grundy set in $G$.
\end{proof}
}

Using the above algorithm, we can compute for each $z \in [k]$ and $v \in X_z$, a family $\C{F}'_{z,v} \subseteq \C{F}_{z,v}$ that satisfies the correctness and the size constraints, in time bounded by $\alpha^{\C{O}(2^{k}+ 1)} \cdot n^{\C{O}(1)}$. Note that $G$ has a Grundy coloring using at least $k$ colors if and only if for some $v \in X_k$, $\C{F}'_{z,v}\neq \emptyset$. This implies a proof of Theorem~\ref{thm:KijGrundy}.

\longversion{	\subsection{Computing Grundy Representatives (Proof of~\Cref{lem:grundy-rep})}\label{sec:Grundy-rep}

The goal of this section is to prove Lemma~\ref{lem:grundy-rep}. Fix any computable function $f: \mathbb{N} \rightarrow \mathbb{N}\setminus \{0\}$. Consider $k \in \mathbb{N}\setminus \{0\}$, a $K_{i,j}$-free graph $G$ on $n$ vertices with a coloring $\chi: V(G)\rightarrow [k]$, where $i \geq j \geq 1$ and $X_z = \chi^{-1}(z)$, for $z \in [k]$. Furthermore, consider $p\in \mathbb{N}$, a vector $\wvec{q} =(q_1,q_2,\cdots, q_k)$, such that $p + \sm{\wvec{q}} \leq f(k)$, and a $p$-family $\C{F}$ of subsets of $V(G)$. Recall that we have \fbox{$\eta = i \cdot f(k) \cdot k$} and \fbox{$\alpha = \alphaval$}, and note that $1 \leq \eta \leq \alpha$. 


We assume that $|\C{F}| \geq 1$, otherwise, we can simply output $\C{F}$. We also assume that the sets in  $\C{F}$ are not repeated, as otherwise, we can remove the repeated entries. We will design a recursive algorithm that will compute the required subset of $\C{F}$. We will first state our base cases, then a correctness invariant for our recursive step, followed by our recursive formula(s).\\

\noindent\textbf{Base Cases:} If $p = 0$, then notice that $\C{F} = \{\emptyset\}$, and thus, $|\C{F}| = 1$, and we simply output $\C{F}$ as the required family. If $\sm{\wvec{q}} = 0$, then, for an arbitrarily chosen set $A \in \C{F}$, we output the set $\{A\}$. The correctness of the base cases is immediate from their descriptions. Moreover, they can be executed in $O(1)$ time.  \\  

\noindent\textbf{Correctness Invariant:} Before we begin computing $\wvec{q}$-Grundy representative for any $p$-family, we assume that we can compute an $\wvec{h}$-Grundy representative for a given $p'$-family $\C{F}'$ with at most $\alpha^{2p' + \sm{\wvec{h}}}$ many sets using at most $\kappa \cdot \alpha^{2p' + \sm{\wvec{h}}} \cdot p \cdot |\C{F}|$ computational steps (where $\kappa \geq 4$ is some fixed constant), whenever: 
\begin{enumerate}
\item $p'=p$ and $\sm{\wvec{h}} < \sm{\wvec{q}}$, or 
\item $p'< p$ and $\sm{\wvec{h}} + p' \leq \sm{\wvec{q}} + p$. 
\end{enumerate}

\noindent\textbf{Recursive Computation:} Now we will discuss the recursive procedure to compute a $\wvec{q}$-Grundy representative for the given $p$-family $\C{F}$. We compute a maximal subfamily $\C{D}\subseteq \C{F}$ of pairwise disjoint sets. We will consider the following cases based on whether $|\C{D}| \leq \eta -1$ or not.



\noindent{\bf Case 1: $|\C{D}| \leq \eta-1$.} This case is a simpler case for us. We let $U_{\C{D}} := \cup_{D \in \C{D}} D$. Note that $|U_{\C{D}}| \leq p \cdot (\eta-1)$, and due to the maximality of $\C{D}$, for any $Y \in \C{F}$, we have $Y \cap U_{\C{D}} \neq \emptyset$. We note that~\Cref{obs:transitive-vec-rep-one-vertex} lies at the crux of the computation for this case. For each $u \in U_{\C{D}}$, let $\C{F}_u := \{Y \in \C{F} \mid u \in Y\}$, and note that $\C{F}_u$ is a $(p,u)$-family and $\C{F}_u -u$ is a $(p-1)$-family. Also, we have $\C{F} = \cup_{u \in U_D} \C{F}_u$. For $u \in U_{\C{D}}$, let $\wvec{q_u}$ be the vector obtained from $\wvec{q}$ by increasing its $\chi(u)$th coordinate by $1$. Compute $\what{\C{F}}_u \qrep{\wvec{q_u}} {\C{F}}_u - u$, where $|\what{\C{F}}_u| \leq \alpha^{2(p-1) + \sm{\wvec{q_u}}} = \alpha^{2p-2+\sm{\wvec{q}} + 1} = \alpha^{2p + \sm{\wvec{q}} -1}$, for each $u \in U_{\C{D}}$. 
For each $u\in U_{\C{D}}$, we also compute $\wtilde{\C{F}}_u \qrep{\wvec{q}} {\C{F}}_u - u$, where $|\wtilde{\C{F}}_u| \leq \alpha^{2(p-1) + \sm{\wvec{q}}} = \alpha^{2p-2+\sm{\wvec{q}}}$. 

Set $\C{F}' := \bigcup_{u\in U_{\C{D}}} \big((\what{\C{F}}_u + u) \cup (\wtilde{\C{F}}_u + u) \big)$, and note that from Observation~\ref{obs:transitive-vec-rep} and~\ref{obs:transitive-vec-rep-one-vertex} we can obtain that $\C{F}'\qrep{\wvec{q}} \C{F}$. Also, we have $\C{F}' \leq 2 \cdot p \cdot (\eta-1) \cdot \alpha^{2p+\sm{\wvec{q}} -1} \leq \alpha^{2p+\sm{\wvec{q}}}$. Notice that the number of computational steps required to compute $\C{F}'$ can be bounded by: $p(\eta-1) \cdot 2 p \cdot \kappa \cdot |\C{F}| \cdot \alpha^{2p + \sm{\wvec{q}}- 1} \leq \kappa \cdot \alpha^{2p + \sm{\wvec{q}}} \cdot p \cdot |\C{F}|$ (recall that $\alpha = \alphaval$). 


\noindent{\bf Case 2: $|\C{D}| \geq \eta$.} Arbitrarily choose a subset $\C{D}^* \subseteq \C{D}$ ($ \subseteq \C{F}$) of size exactly $\eta$, and let $U = \cup_{D\in \C{D}^*} D$. Before proceeding further, we intuitively say what our algorithm will do next. For almost all sets $B$ of size $\wvec{q}$ that fits a set in $\C{F}$, we will obtain that $\C{D}^*$ already contains a set that fits $B$. We will obtain a small set $S$, such that any set $B$ of size $\wvec{q}$ that is not already taken care of by $\C{D}^*$, will contain an element from $S$. Once we have this set $S$, we will compute an appropriate representative for them exploiting the small size of $S$. We now formalize the above. 

Consider any $z \in [k]$, and let $U_z = X_z \cap U$. A vertex $v \in X_z$ is a {\em $z$-heavy} vertex if $|N_{G}(v) \cap U_z| \geq \degval$. (In the above, $v$ could also possibly belong to the set $U_z$.) Let $S_z$ be the set of all $z$-heavy vertices in $G$, and $S = \cup_{z' \in [k]} S_{z'}$. Next, we argue that $\C{D}^*$ already fits sets $B$ of size $\wvec{q}$ that do not contain vertices from $S$ (if there is a set in $\C{F}$ that fits $B$). 

\begin{lemma}\label{clm:Dstr-almost-qresp}
Consider any $B\subseteq V(G)$ of size $\wvec{q}$ for which there is $A \in \C{F}$ that fits $B$. Then, one of the following holds: i) there is $A' \in \C{D}^*$ such that $A'$ fits $B$, or ii) $B \cap S \neq \emptyset$.  
\end{lemma}
\begin{proof}
Consider a set $B$ that satisfies the premise of the claim, and note that $|B| = \sm{\wvec{q}}$ (and also $|B| = \wvec{q}$). Recall that $\C{D}^*$ is a set of pairwise disjoint sets from $\C{F}$ with exactly $\eta$ sets. 
Let $\C{D}^* = \{A_1,A_2,\cdots, A_\eta\}$. Suppose that there is no $A \in \C{D}^*$ that fits $B$ (otherwise, the claim trivially follows). This means that, for each $j \in [\eta]$, there is $z_j \in [k]$, $v_j \in A_j \cap X_{z_j}$ and $u_j \in B\cap X_{z_j}$, such that $\{v,u\} \in E(G)$ (and we work with such arbitrarily fixed vertices $v_j,u_j$, for each $j \in [\eta]$). For distinct $j,j' \in [\eta]$, $v_j \neq v_{j'}$, as sets in $\C{D}^*$ are pairwise vertex disjoint, although $u_j$ could possibly be the same as $u_{j'}$. For any $j \in [\eta]$, $v_j \neq u_j$ as $G$ is a simple graph and $\{v_i,u_i\} \in E(G)$. For distinct $j,j' \in [\eta]$, $v_j \neq u_{j'}$, as otherwise, $u_j,u_{j'} \in B \cap X_{z_j}$ with $\{u_j,u_{j'}\} \in E(G)$, contradicting that $B$ fits with some $A\in \C{F}$. Let $B' = \{u_j \mid j \in [\eta]\}$ and $D' = \{v_j \mid j \in [\eta]\}$. From the discussions above we have, $|D'| = \eta$, $D' \cap B' =\emptyset$, each vertex in $D'$ has a neighbor in $B'$, and each vertex in $B'$ has a neighbor in $D'$. Recall that $\eta = i \cdot f(k) \cdot k \geq \etaval{\wvec{q}}$. For each $z \in [k]$, let $B'_z = B' \cap X_z$. By the Pigeonhole principle, there is $z^*\in [k]$, and indices $j'_1,j'_2, \cdots, j'_{\degval \cdot \sm{\wvec{q}}} \in [\eta]$, such that for every $s \in [\degval \cdot \sm{\wvec{q}}]$, $u_{j'_s} \in B'_{z^*}$. As $|B'| \leq |B| = \sm{\vec{q}}$, again by the Pigeonhole principle, there are distinct indices $j_1,j_2, \cdots, j_\degval \in [\eta]$, such that $u_{j_s} = u_{j_{s'}}$, for every $s,s' \in [i]$, and $u_{j_1} \in B_{z^*}$. Note that $u_{j_1} \in B_{z^*}$ is adjacent to each of the vertices in $\{v_{j_1},v_{j_2}, \cdots, v_{j_\degval}\} \subseteq D' \cap X_{z^*}$. Thus we can conclude that $u_{j_1} \in B$ is a $z^*$-heavy vertex, and thus, $B \cap S \neq \emptyset$.  
\end{proof}

Next, we bound the size of $S$ by exploiting the $K_{i,j}$-freeness of the graph $G$.

\begin{lemma}\label{clm:size-bound-S}
For each $z \in [k]$, $|S_z\setminus U_z| < \szsize$, and thus $|S| \leq k \cdot \szsize + p \cdot \eta$.   
\end{lemma}
\begin{proof}
Towards a contradiction suppose there is $z \in [k]$, such that $|S_z\setminus U_z| \geq \szsize$. Let $S'_z = S_z\setminus U_z$. Let $H_z$ be the bipartite graph with bipartition $(U_z, S'_z)$, where for $u \in U_z$ and $v \in S'_z$, we add the edge $\{u,v\}$ to $E(H_z)$ if and only if $\{u,v\} \in E(G)$. Note that $H_z$ is a subgraph of $G$ so it is must be $K_{i,j}$-free, and thus, $K_{i,i}$-free (recall that $i\geq j$). 
As each $v \in S'_z$ is $z$-heavy, we can obtain that $\degr_{H_z}(v) \geq \degval$, and $\sum_{u \in U_z} \degr_{H_z}(u) = \sum_{v \in S'_z} \degr_{H_z}(v) \geq \degval \cdot |S'_z|$. Now we will iteratively construct a sequence of distinct vertices $u_1,u_2, \cdots, u_i \in U_z$ and a sequence of sets $S'_z \supseteq W_1 \supseteq W_2 \supseteq \cdots \supseteq W_i$, where we want to ensure that: i) for each $j' \in [i]$, the vertex $u_{j'}$ is adjacent to all the vertices in $W_{j'}$, and ii) for each $j' \in [i]$, $|W_{j'}| \geq (p \cdot \eta)^{i-j'+1}$. As $|S'_z| \geq \szsize$, $|U_z| \leq p \cdot \eta$ and $N_{H_z}(v) \cap U_z \neq \emptyset$, for each $v \in S'_z$, by the Pigeonhole principle, there exists $u_1 \in U_z$ that is adjacent to at least $(p \cdot \eta)^{i}$ vertices in $S'_z$. Set $W_1 := N_{H_z}(u_1)$, and note that $|W_1| \geq (p \cdot \eta)^{i}$. 

Now suppose for some $2 \leq \what{j} \leq i$, we have constructed the vertex $u_{j'}$ and the set $W_{j'}$ satisfying the required properties, for each $j' \in [\what{j}-1]$. Now we construct the vertex $u_{\what{j}}$ and the set $W_{\what{j}}$ as follows. We have $|W_{{\what{j}}-1}| \geq (p \cdot \eta)^{i-{\what{j}} + 2}$, where $W_{{\what{j}}-1} \subseteq S'_z$. Each vertex $w \in W_{{\what{j}}-1}$ has at least $\degval$ neighbors in $U_z$, and $w$ is adjacent to all the vertices in $\{u_1, u_2 \cdots, u_{{\what{j}}-1}\}$. Thus, each $w \in W_{{\what{j}}-1}$ has at least $i-({\what{j}}-1) \geq 1$ (note that ${\what{j}} \leq i$) neighbors in $U_z \setminus \{u_1, u_2 \cdots, u_{j-1}\}$. Again by the Pigeonhole principle, there exists $u_{\what{j}} \in U_z \setminus \{u_1, u_2 \cdots, u_{j-1}\}$, that is adjacent to at least $|W_{{\what{j}}-1}|/ |U_z \setminus \{u_1, u_2 \cdots, u_{j-1}\}| \geq |W_{{\what{j}}-1}|/ (p \cdot \eta) \geq (p \cdot \eta)^{i-{\what{j}} + 1}$ vertices in $W_{{\what{j}}-1}$. Set $W_{{\what{j}}} = N_G(u_{\what{j}}) \cap W_{{\what{j}}-1}$, and note that $|W_{\what{j}}| \geq (p \cdot \eta)^{i-{\what{j}} + 1}$. This finishes the construction of $u_{\what{j}}$ and $W_{\what{j}}$, for each ${\what{j}} \in [i]$, with the desired properties. Note that each $w \in W_i$ is adjacent to each vertex in $\{u_1,u_2,\cdots, u_i\}$ and $|W_i| \geq p \cdot \eta \geq i$. Moreover, by construction, $\{u_1,u_2,\cdots, u_i\} \cap W_i =\emptyset$ and $|\{u_1,u_2,\cdots, u_i\}| = i$. From the above discussions we can conclude that $G[W_i\cup \{u_1,u_2,\cdots, u_i\}]$ contain $K_{i,i}$ as a subgraph, which is a contradiction. Thus we conclude that $|S_z\setminus U_z| < \szsize$
\end{proof} 

From~\Cref{clm:Dstr-almost-qresp}, we now need to construct $\C{D}' \subseteq \C{F}$, such that each set $B$ of size $\wvec{q}$ that contains a vertex from $S$ and fits a set from $\C{F}$, we have a set in $\C{D}'$ that fits $B$. Note that once we have such a set $\C{D}'$ with the guarantee that $|\C{D}^*\cup\C{D}'| \leq \alpha^{2p +\sm{\wvec{q}}}$, then we can output $\C{D}^*\cup\C{D}'$ as the desired $\wvec{q}$-Grundy representative for $\C{F}$. We now focus on constructing such a $\C{D}'\subseteq \C{F}$.

Let $S'$ be the subset of $S$ containing all those $s \in S$ for which the $\chi(s)$th coordinate of $\wvec{q}$ is at least $1$. For each $s \in S'$, let $\C{F}_s := \{A \in \C{F} \mid A \cup \{s\} \mbox{ is $\chi$-independent}\}$, and let $\wvec{q_s}$ be the vector obtained from $\wvec{q}$ by decreasing its $\chi(s)$th coordinate by $1$, and thus, we have $0 \leq \sm{\wvec{q_s}} < \sm{\wvec{q}}$. We compute $\C{F}'_s \qrep{\wvec{q_s}} \C{F}_s$ with at most $\alpha^{2p+ \sm{\wvec{q}} -1}$ sets (see the correctness invariant). Set $\C{F}' := \C{D}^* \cup (\cup_{s\in S'} \C{F}'_s)$, and note that $|\C{F}'| \leq \eta + \alpha^{2p+ \sm{\wvec{q}} -1} \cdot |S| \leq \eta + \alpha^{2p+ \sm{\wvec{q}} -1} \cdot (k \cdot (p \eta)^{i+1} + p \eta) \leq \alpha^{2p+ \sm{\wvec{q}} -1} \cdot 3k \cdot (p\eta)^{i+1} \leq \alpha^{2p+ \sm{\wvec{q}}}$. We return $\C{F}'$ as the $\wvec{q}$-Grundy representative for $\C{F}$, and the next lemma proves the correctness of this step. 

\begin{lemma}\label{claim-corretrecur-qrep}
$\C{F}' \qrep{\wvec{q}} \C{F}$, and we can compute $\C{F}'$ using at most $\kappa \cdot \alpha^{2p + \sm{\wvec{q}}} \cdot p \cdot |\C{F}|$ computational steps. 
\end{lemma}
\begin{proof}
Consider any $B\subseteq V(G)$ of size $\wvec{q}$, for which there is $A \in \C{F}$, such that $A$ fits $B$. If $B \cap S = \emptyset$, then from~\Cref{clm:Dstr-almost-qresp}, there is $A' \in \C{D}^*\subseteq \C{F}'$, such that $A'$ fits $B$. Now we consider the case when $B \cap S \neq \emptyset$, and choose an arbitrary vertex $s \in B \cap S$. Note that $A \in \C{F}_s$ as $A$ fits $B$, and thus, $A \cup \{s\}$ is $\chi$-independent. Also, note that $A$ fits $B':=B \setminus \{s\}$, and $|B'| = \wvec{q_s}$. As $\C{F}'_s \qrep{\wvec{q_s}} \C{F}_s$ and $A \in \C{F}_s$ fits $B'$ of size $\wvec{q_s}$, there exists $\what{A} \in \C{F}'_s$ that fits $B'$. As $\what{A} \in \C{F}'_s \subseteq \C{F}_s$, we can obtain that $\what{A}\cup \{s\}$ is $\chi$-independent. Moreover, we also have that $B$ is $\chi$-independent. From the above discussions, we can conclude that $\what{A}$ fits $B$. Hence we conclude that $\C{F}' \qrep{\wvec{q}} \C{F}$. 

The number of computational steps required to compute $\C{F}'$ can be bounded by:

$$|\C{F}| \cdot p + (k \cdot \szsize + p \cdot \eta) \cdot (\kappa \cdot \alpha^{2p + \sm{\wvec{q}} - 1} \cdot p \cdot |\C{F}|)$$

$$\leq |\C{F}| \cdot p + 2k \cdot \szsize \cdot (\kappa \cdot \alpha^{2p + \sm{\wvec{q}} - 1} \cdot p \cdot |\C{F}|)$$

 As $\alpha = \alphaval$, we can obtain that $2\alpha/3 = 2k \szsize$. Also, recall that $\kappa \geq 4$. Thus, we have:
 
 $$|\C{F}| \cdot p + 2k \cdot \szsize \cdot \kappa \cdot \alpha^{2p + \sm{\wvec{q}} - 1} \cdot p \cdot |\C{F}| \leq |\C{F}| \cdot p + \frac{2\alpha}{3} \cdot \kappa \cdot \alpha^{2p + \sm{\wvec{q}} - 1} \cdot p \cdot |\C{F}|$$
 
 
  $$\leq |\C{F}| \cdot p \cdot \alpha^{2p + \sm{\wvec{q}}} \cdot (\frac{1}{\alpha} + \frac{2\kappa}{3}) \leq \kappa \cdot |\C{F}| \cdot p \cdot \alpha^{2p + \sm{\wvec{q}}}$$
  
This concludes the proof.
\end{proof}

The correctness of the above algorithm and the desired running time follows from the correctness of the base cases, the correctness of Case 1 (from Observation~\ref{obs:transitive-vec-rep} and~\ref{obs:transitive-vec-rep-one-vertex}), and their runtime analysis, and Lemma~\ref{claim-corretrecur-qrep}. This gives us proof of~\Cref{lem:grundy-rep}. 

}


\bibliography{References}

\end{document}